\documentclass[reprint,aps,twoside,prd,superscriptaddress,nofootinbib,longbibliography]{revtex4-2}
\pdfoutput=0 
\usepackage{graphicx}
\usepackage[T1]{fontenc}
\usepackage{amssymb}
\usepackage{amsmath}
\usepackage{amsthm}
\usepackage{tensor,mathtools}
\usepackage{nicefrac}
\usepackage{dcolumn}
\usepackage{cancel} 
\usepackage{physics}
\usepackage{makecell}

\usepackage{enumerate}

\usepackage[version=4]{mhchem}

\usepackage[dvipsnames]{xcolor} 

\usepackage{tikz,pgfplots} 
\usetikzlibrary{decorations.pathreplacing} 
\usetikzlibrary{patterns}
\usetikzlibrary{arrows.meta,arrows}

\newtheorem{proposition}{Proposition}
\newtheorem*{theorem*}{Theorem}

\usepackage{siunitx} 
\DeclareSIUnit\year{yr}
\DeclareSIUnit \parsec {pc}
\usepackage{comment}

\usepackage{url}

\usepackage{aasmacros} 

\newcommand\numberthis{\addtocounter{equation}{1}\tag{\theequation}}

\usepackage[colorlinks]{hyperref}
\hypersetup{
     breaklinks=true,
    pdfstartview={FitH},    
    colorlinks=true,       
    linkcolor=blue,          
    citecolor=red,        
    filecolor=magenta,      
    urlcolor=blue,           
    anchorcolor=green,      
    linktocpage=true
}
\usepackage{cleveref}

\newcommand{\afffias}{Frankfurt Institute for Advanced Studies (FIAS), Ruth-Moufang-Str.~1, 60438 Frankfurt am Main, Germany}
\newcommand{\affgoethe}{Physics Department, Goethe University, Max-von-Laue-Str.~1, 60438 Frankfurt am Main, Germany}
\newcommand{\affcharles}{Faculty of Mathematics and Physics,
Charles University, Ke Karlovu 3, 121 16 Praha 2, Czech Republic}

\date{\today}

\begin{document}

\title{Toward Singularity Theorems with Torsion}

\author{A.~van de Venn}
\email{venn@fias.uni-frankfurt.de}
\affiliation{\afffias}
\affiliation{\affgoethe}

\author{U.~Agarwal}
\email{ujjwal.agarwal@mff.cuni.cz}
\affiliation{\affcharles}

\author{D.~Vasak}
\email{vasak@fias.uni-frankfurt.de}
\affiliation{\afffias}

\begin{abstract}
This study examines the formulation of a singularity theorem for timelike curves including torsion, and establishes the foundational framework necessary for its derivation.
We begin by deriving the relative acceleration for an arbitrary congruence of timelike curves. The resulting ``deviation equation'' offers an alternative pathway to the well-known Raychaudhuri equation with torsion. Conjugate points are then introduced and analyzed in relation to the behavior of the scalar expansion. Together with the sensible requirement of hypersurface orthogonality, the Raychaudhuri equation is examined for several specific cases of torsion that are prominent in the literature.
Our findings indicate that a totally antisymmetric torsion tensor does not influence the behavior of the congruence of timelike curves. Finally, we formulate a singularity theorem for timelike curves and highlight the critical requirement of non-autoparallel curves.
\end{abstract}

\maketitle

\section{Introduction} 
Despite its tremendous success, Einstein's General Relativity (GR) is plagued by the appearance of singularities, as 
elucidated by the renowned Penrose-Hawking singularity
theorems \cite{hawking1975large,Penrose_Sing,Hawking_Bondi}.
Singularities are physically undesirable as they often signify
divergences in measurable quantities and the violation of
physical laws.
While there are ways to avoid, e.g., the Big Bang singularity
within the framework of GR \cite{Carballo-Rubio:2024rlr},
we are going to adopt a different approach here.

In Metric-Affine Gravity (MAG) the linear connection is
promoted to an independent dynamic variable and provides,
besides the pseudo-Riemannian metric, additional degrees of 
freedom. These additional degrees of freedom are referred to
as ``torsion'' and ``nonmetricity''.
Geometrically, torsion measures the non-closure of a parallelogram formed by parallel transporting two infinitesimal vectors along each other \cite{BeltranJimenez:2019esp}. Physically, torsion is believed
to be closely related to the spin-density of matter \cite{Hehl1,SHAPIRO2002113} and is able to account for the
effects of Dark Energy in the late Universe \cite{vandeVenn:2022gvl,Kirsch:2023iwd,Benisty:2021sul,Vasak:2022pii}.
Nonmetricity on the other hand causes lengths of vectors and
angles between vectors to change when performing a parallel 
transport. This geometrical picture is however not 
complemented by any promising physical significance yet.
Therefore we are going to restrict ourselves to the study
of a connection including torsion while neglecting any
nonmetricity contributions, i.e., we assume a metric
compatible connection.

The framework for establishing and proving singularity theorems is based on the Raychaudhuri equation. It 
determines how a ball of events within the congruence of 
curves contracts/expands as we move along the flow of the curves. A crucial aspect of the Raychaudhuri equation involves the imposition of energy conditions (ECs). These ECs are ad-hoc conditions based on our current interpretation of gravitational forces.
By applying these ECs to the conventional Raychaudhuri equation, it can be concluded that a collection of vorticity-free geodesics will converge to a focal point, ultimately leading to the formation of a spacetime singularity. However, the presence of torsion modifies the kinematics, potentially leading to different outcomes. Therefore, the study of kinematics in spacetimes with torsion provides insights into modified ECs, which may enhance our understanding of gravity and potentially avoid the formation of singularities \cite{Esposito:1990ub}.

The modified ECs we will derive are presented in what is referred to as the geometric form \cite{Curiel:2014zba}. Only in the context of the Einstein equations, this geometric form is directly connected to the energy-momentum tensor, as well as to the thermodynamic properties of matter, including energy density, pressure, and the equation of state. However, except for our later discussion of the Weyssenhoff fluid, we do not assume a specific gravitational model in this work. Consequently, a definitive physical energy interpretation for these conditions is not provided within this work.

This paper is organised as follows: \Cref{theor_frame}
introduces the basic definitions of our framework and establishes the conventions used throughout the paper.
The deviation equation for timelike curves is derived in
\Cref{dev_eq_sect} and culminates in the Raychaudhuri eq.
with torsion in \Cref{Raychaudh_sect}. We discuss the relation
between the conventional kinematic variables and their geometric counterparts in \Cref{Geom_var_sect}. 
The notion of conjugate points is generalised to the case
of arbitrary timelike curves in \Cref{Conj_points_sec}
and allows us, together with hypersurface orthogonality
(\Cref{Hypersur_orthog_sec}), to 
investigate particular torsion forms in \Cref{Case_stud_sec}.
Finally, we present and prove a singularity theorem for arbitrary timelike curves with torsion in \Cref{Sing_thm_sec}, and conclude the paper with a summary of our findings in \Cref{sect:Concl}.

Throughout this paper we employ natural units, in which
$\hbar=c=1$.
Furthermore, the metric is assumed to have signature
$(-+++)$, following the standard convention in \cite{Misner1973}.

\section{Theoretical framework}
\label{theor_frame}
In metric compatible MAG the connection is given by
\begin{equation}
    \Gamma\indices{^{\lambda}_{\mu\nu}} = \mathring{\Gamma}\indices{^{\lambda}_{\mu\nu}} + K\indices{^{\lambda}_{\mu\nu}},
    \label{affine_conn}
\end{equation}
where 
\begin{equation}
    \mathring{\Gamma}\indices{^{\lambda}_{\mu\nu}} =
    \frac{1}{2}g\indices{^{\lambda\rho}}\left(\partial\indices{_\nu}g\indices{_{\mu\rho}} + \partial\indices{_\mu}g\indices{_{\nu\rho}}
    - \partial\indices{_\rho}g\indices{_{\mu\nu}}\right)
\end{equation}
are the Christoffel symbols with $g\indices{_{\mu\nu}}$ being the metric, and 
\begin{equation}
    K\indices{^{\lambda}_{\mu\nu}} \coloneqq 
    g\indices{^{\lambda\rho}}\left(S\indices{_{\rho\mu\nu}} - S\indices{_{\mu\rho\nu}} - S\indices{_{\nu\rho\mu}}\right)
\end{equation}
is the contortion tensor defined in terms of the torsion
tensor $S\indices{^{\lambda}_{\mu\nu}} \coloneqq \Gamma\indices{^{\lambda}_{[\mu\nu]}}$.

Curvature is encoded within the Riemann-Cartan tensor
\begin{equation}
    R\indices{^\lambda_{\sigma\mu\nu}} \coloneqq \partial\indices{_\mu}\Gamma\indices{^{\lambda}_{\sigma\nu}} - \partial\indices{_\nu}\Gamma\indices{^{\lambda}_{\sigma\mu}} + \Gamma\indices{^{\lambda}_{\rho\mu}}\Gamma\indices{^{\rho}_{\sigma\nu}} - \Gamma\indices{^{\lambda}_{\rho\nu}}\Gamma\indices{^{\rho}_{\sigma\mu}},
\end{equation}
which may be split into its Levi-Civita contribution, 
denoted by a ring, and its torsional part using \eqref{affine_conn}:
\begin{align*}
    R\indices{^\lambda_{\sigma\mu\nu}} = &\mathring{R}\indices{^\lambda_{\sigma\mu\nu}}
    + \mathring{\nabla}\indices{_\mu}K\indices{^\lambda_{\sigma\nu}} - \mathring{\nabla}\indices{_\nu}K\indices{^\lambda_{\sigma\mu}} \\&+ K\indices{^\lambda_{\rho\mu}}K\indices{^\rho_{\sigma\nu}}
    - K\indices{^\lambda_{\rho\nu}}K\indices{^\rho_{\sigma\mu}}
    \numberthis.
\end{align*}

Covariant derivatives act on an arbitrary $(1,1)$ tensor $T$ as
\begin{equation}
    \nabla\indices{_\mu}T\indices{^\alpha_\beta} = 
    \partial\indices{_\mu}T\indices{^\alpha_\beta} + 
    \Gamma\indices{^\alpha_{\rho\mu}}T\indices{^\rho_\beta}
    - \Gamma\indices{^\rho_{\beta\mu}}T\indices{^\alpha_\rho},
\end{equation}
and accordingly on higher/lower rank tensors. As before, this 
expression may be split into
\begin{equation}
    \nabla\indices{_\mu}T\indices{^\alpha_\beta} = 
    \mathring{\nabla}\indices{_\mu}T\indices{^\alpha_\beta} + 
    K\indices{^\alpha_{\rho\mu}}T\indices{^\rho_\beta}
    - K\indices{^\rho_{\beta\mu}}T\indices{^\alpha_\rho}.
\end{equation}

With torsion present, the notions of autoparallel and geodesic
curves differ in general. A geodesic is a curve that 
minimises/maximises the path length
\begin{equation}
    S = \int\sqrt{-g\indices{_{\mu\nu}}\frac{\mathrm{d}x\indices{^\mu}}{\mathrm{d}\lambda}\frac{\mathrm{d}x\indices{^\nu}}{\mathrm{d}\lambda}}\mathrm{d}\lambda,
\end{equation}
where $x\indices{^\mu}(\lambda)$ denotes the curve in local 
coordinates and $\lambda$ is the curve parameter. The path
length $S$ is extremal if $\delta S = 0$, which leads to the
geodesic equation
\begin{equation}
    \frac{\mathrm{d}^2 x\indices{^\mu}}{\mathrm{d}\lambda^2}
    + \mathring{\Gamma}\indices{^\mu_{\alpha\beta}}\frac{\mathrm{d}x\indices{^\alpha}}{\mathrm{d}\lambda}\frac{\mathrm{d}x\indices{^\beta}}{\mathrm{d}\lambda} = 0.
\end{equation}
Notice that only the Levi-Civita part of the connection contributes to the geodesic equation.
An autoparallel curve on the other hand generalises the notion
of a straight line to curved space. Let $X\indices{^\mu}$
denote the tangent vector field to a curve, then the condition
for this curve to be an autoparallel reads $X\indices{^\alpha}\nabla\indices{_\alpha}X\indices{^\mu} = 0$. In local coordinates with $X\indices{^\mu} = \mathrm{d}x\indices{^\mu}/\mathrm{d}\lambda$, this condition
is equivalent to
\begin{equation}
    \frac{\mathrm{d}^2 x\indices{^\mu}}{\mathrm{d}\lambda^2}
    + \Gamma\indices{^\mu_{\alpha\beta}}\frac{\mathrm{d}x\indices{^\alpha}}{\mathrm{d}\lambda}\frac{\mathrm{d}x\indices{^\beta}}{\mathrm{d}\lambda} = 0.
\end{equation}
Evidently now the full connection, more specifically the symmetric part $\Gamma\indices{^\mu_{(\alpha\beta)}}$, has
to be taken into account. Thus torsion enters the autoparallel equation in general.
Only for a totally antisymmetric torsion tensor (or a vanishing torsion) do the notions of geodesics and 
autoparallels coincide.

\section{Deviation Equation}
\label{dev_eq_sect}
Let $\mathcal{M}$ be an $n$ dimensional Lorentzian manifold
with an independent metric compatible connection $\Gamma\indices{^\lambda_{\mu\nu}}$
and $\gamma(s)$ be a curve in $\mathcal{M}$ with tangent vectors
\begin{equation}
    {\hat{V}}_{\gamma(s)}^{\mu} \coloneqq \frac{\mathrm{d} {\gamma}\indices{^{\mu}}}{\mathrm{d} s}(s),
\end{equation}
defined at every point along the curve.
Let further $U\indices{^{\mu}}$ be a normalised timelike vector field,
i.e., $g(U,U) = U\indices{^{\mu}}U\indices{_{\mu}} = -1$ at every point where $U\indices{^{\mu}}$ is defined, see Fig.\,\ref{fig.2.1}.
We may create a one-parameter family of non-intersecting curves by moving each point on
$\gamma(s)$ along the integral curves of $U\indices{^{\mu}}$.
This family shall be denoted by $\gamma_s (\tau)$
\cite{hawking1975large}. In other words, for fixed $s$, the quantity 
$\gamma_s (\tau)$ constitutes a timelike curve parametrised by the 
affine parameter $\tau$ which coincides with the proper time.
We thus have
\begin{equation}
    U_{{\gamma_s}(\tau)}^{\mu} = \frac{\partial {\gamma_s}\indices{^{\mu}}}{\partial \tau}(\tau)
\end{equation}
at any point within the family of curves,
where ${\gamma_s}\indices{^{\mu}}(\tau)$ are the coordinates
of the curves in a local frame. We can extend the 
vector field $\hat{V}\indices{^{\mu}}$ across the whole family 
$\gamma_s (\tau)$ by pushing all $\hat{V}\indices{^{\mu}}$ along the integral curves to an arbitrary
parameter value of $\tau$.
For that matter we define the flow $\phi_{\tau}$ of $U\indices{^\mu}$ as usual by
\begin{equation}
    \phi_{\tau}(\gamma_s(0)) \coloneqq \gamma_s(\tau),
\end{equation}
where $\gamma_s(0) = \gamma(s)$.
It is well known that  the $\phi_{\tau}$ constitute a 
one-parameter family of diffeomorphisms. We then extend
$\hat{V}\indices{^{\mu}}$ to
\begin{equation}
    {\tilde{V}}_{\gamma_s (\tau)}
     \coloneqq (\phi_{\tau})_* \left({\hat{V}}_{\gamma(s)}\right),
     \label{ext_def}
\end{equation}
where $(\phi_{\tau})_*$ denotes the pushforward of the flow.
In other words
\begin{equation}
    {\tilde{V}}_{\gamma_s (\tau)}^\mu = \frac{\partial {\gamma_s}\indices{^{\mu}}}{\partial s} (\tau),
    \label{ext_res}
\end{equation}
see \Cref{sect_flow_calc} for details.
Our goal is to determine how the spacelike separation between neighbouring timelike curves changes as we move along them.
This is essentially described by the change of $\tilde{V}\indices{^{\mu}}$ along the curves. In general
however, $\tilde{V}\indices{^{\mu}}$ may have components
parallel to $U\indices{^{\mu}}$ which we do not care about.
A remedy is to project $\tilde{V}\indices{^{\mu}}$ onto 
a vector subspace orthogonal to $U\indices{^{\mu}}$ first
and consider the change of this orthogonal component only.
The projection to the orthogonal subspace is achieved with
the projection tensor
\begin{equation}
    P\indices{^\mu_\nu}
    \coloneqq \delta\indices{^\mu_\nu}
    + U\indices{^{\mu}}U\indices{_{\nu}}.
\end{equation}
We thus define the new \textit{deviation vector field},
\begin{equation}
    V\indices{^{\mu}} \coloneqq P\indices{^\mu_\nu}
    \tilde{V}\indices{^{\nu}},
\end{equation}
orthogonal to $U\indices{^{\mu}}$.

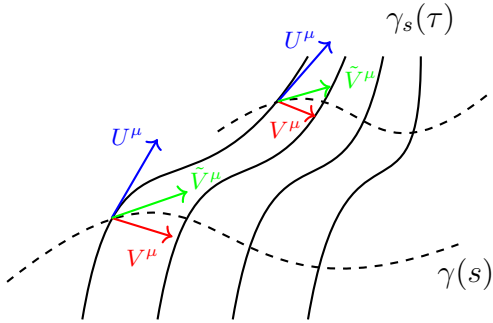
\begin{figure}[h]
	\begin{tikzpicture}
		\draw[thick] (4,-0.5) .. controls (4.5,2.5) and (5.5,0.5) .. (7,3);
        \draw[thick] (5,-0.5) .. controls (5.5,2.5) and (6.5,0.5) .. (7.5,3);
        \draw[thick] (6,-0.5) .. controls (6.5,2.5) and (7.5,0.5) .. (8,3);
        \draw[thick] (7,-0.5) .. controls (7.5,2.5) and (8.5,0.5) .. (8.5,3)
        node[xshift=0cm,yshift=0.5cm]{\large{$\gamma_s (\tau)$}};
        \draw[thick,dashed] (3,0) .. controls (6,2.4) and (5.5,-0.8) .. (9,0.5)
        node[xshift=0.1cm,yshift=-0.4cm]{\large{$\gamma (s)$}};
        \draw[thick,dashed] (5.8,2) .. controls (7.8,3.35) and (7.5,1) .. (9.4,2.5);
        \draw[thick,->,color = blue] (4.39,0.85) -- (5,1.9)
        node[xshift=-0.4cm]{$U\indices{^\mu}$};
        \draw[thick,->,color = red] (4.39,0.85) -- (5.2,0.6)
        node[xshift=-0.4cm,yshift=-0.3cm]{$V\indices{^\mu}$};
        \draw[thick,->,color = green] (4.39,0.85) -- (5.4,1.2)
        node[xshift=0.25cm,yshift=0.15cm]{$\tilde{V}\indices{^\mu}$};
        \draw[thick,->,color = blue] (6.6,2.4) -- (7.3,3.2)
        node[xshift=-0.4cm]{$U\indices{^\mu}$};
        \draw[thick,->,color = red] (6.6,2.4) -- (7.1,2.2)
        node[yshift=-0.2cm,xshift=-0.4cm]{$V\indices{^\mu}$};
        \draw[thick,->,color = green] (6.6,2.4) -- (7.3,2.6)
        node[yshift=0.1cm,xshift=0.4cm]{$\tilde{V}\indices{^\mu}$};
	\end{tikzpicture}
	\caption{The family of timelike curves $\gamma_s (\tau)$ with 
    tangent vector field $U\indices{^\mu}$ and deviation vector field $V\indices{^\mu}$}
	\label{fig.2.1}
\end{figure}

To determine how $V\indices{^{\mu}}$ changes along the curves, we first note that 
the Lie derivative of $\tilde{V}$ along $U$ vanishes
\begin{equation}
    \mathcal{L}_U \tilde{V} \overset{\text{def}}{=} \left.\frac{\mathrm{d}}{\mathrm{d}\tau}\left((\phi_\tau)^* (\tilde{V})\right)\right\rvert_{\tau=0} = 0,
\end{equation}
where $(\phi_\tau)^* = (\phi_\tau^{-1})_*$ denotes the pullback,
since 
\begin{equation}
    (\phi_\tau)^* (\tilde{V}) = {\hat{V}}
\end{equation}
is independent of $\tau$.
But $\mathcal{L}_U \tilde{V} = [U,\tilde{V}]$ and
therefore we have that 
\begin{align*}
    0 &= [U,\tilde{V}]\indices{^\mu}\\
    &= U\indices{^\alpha}\partial
    \indices{_\alpha} \tilde{V}\indices{^\mu} - 
    \tilde{V}\indices{^\alpha}\partial
    \indices{_\alpha} U\indices{^\mu}\\
    &= U\indices{^\alpha}\nabla
    \indices{_\alpha} \tilde{V}\indices{^\mu} - 
    \tilde{V}\indices{^\alpha}\nabla
    \indices{_\alpha} U\indices{^\mu} + 2 U\indices{^\alpha}\tilde{V}\indices{^\beta}\Gamma\indices{^\mu_{
    [\alpha\beta]}}\\
    &= U\indices{^\alpha}\nabla
    \indices{_\alpha} \tilde{V}\indices{^\mu} - 
    \tilde{V}\indices{^\alpha}\nabla
    \indices{_\alpha} U\indices{^\mu} + 2 U\indices{^\alpha}\tilde{V}\indices{^\beta} S\indices{^\mu_{
    \alpha\beta}}, \numberthis
    \label{commutator1}
\end{align*}
and hence
\begin{equation}
    U\indices{^\alpha}\nabla
    \indices{_\alpha} \tilde{V}\indices{^\mu} = 
    \tilde{V}\indices{^\alpha}\nabla
    \indices{_\alpha} U\indices{^\mu} + 2 U\indices{^\alpha}\tilde{V}\indices{^\beta} S\indices{^\mu_{
    \beta\alpha}}.
    \label{cov_exch}
\end{equation}
To proceed, note that 
\begin{equation}
    U\indices{^\beta}\nabla\indices{_\beta}
    P\indices{^\rho_\lambda}=
    A\indices{^\rho}U\indices{_\lambda} + A\indices{_\lambda}
    U\indices{^\rho},
    \label{proj_id}
\end{equation}
where $A\indices{^\mu} \coloneqq U\indices{^\beta}\nabla\indices{_\beta} U\indices{^\mu}$
is the autoparallel acceleration.
Now multiplying \eqref{cov_exch} by $P\indices{^\nu_\mu}$
and summing over $\mu$, we find, after some algebra and 
with \eqref{proj_id}, that
\begin{align*}
    U\indices{^\alpha}\nabla
    \indices{_\alpha} V\indices{^\mu} =& 
    A\indices{^\mu}U\indices{_\alpha}\tilde{V}\indices{^\alpha}
    + U\indices{^\mu}A\indices{_\alpha}\tilde{V}\indices{^\alpha}+
    \tilde{V}\indices{^\alpha}\nabla
    \indices{_\alpha} U\indices{^\mu} \\&+ 2 U\indices{^\alpha}\tilde{V}\indices{^\beta} 
    P\indices{^\mu_\nu}
    S\indices{^\nu_{
    \beta\alpha}}.\numberthis
    \label{cov_exch_V}
\end{align*}
After projecting and using
$A\indices{^\lambda}U\indices{_\alpha}
    \tilde{V}\indices{^\alpha} + \tilde{V}\indices{^\alpha}\nabla
    \indices{_\alpha} U\indices{^\lambda} = 
    V\indices{^\alpha}\nabla
    \indices{_\alpha} U\indices{^\lambda}$
this reads
\begin{equation}
    P\indices{^\lambda_\mu}U\indices{^\alpha}\nabla
    \indices{_\alpha} V\indices{^\mu} =
    V\indices{^\alpha}\nabla
    \indices{_\alpha} U\indices{^\lambda} + 2 U\indices{^\alpha}\tilde{V}\indices{^\beta} 
    P\indices{^\lambda_\nu}
    S\indices{^\nu_{
    \beta\alpha}}
    \label{cov_exch_V_proj}
\end{equation}
and tells us how the spacelike separation of neighbouring curves changes along the curves.

We may go one step further and determine how 
$P\indices{^\lambda_\mu}U\indices{^\alpha}\nabla
    \indices{_\alpha} V\indices{^\mu}$
changes along the curves, which leads us to the 
\textit{acceleration} of nearby timelike curves
\begin{equation}
    a\indices{^\mu} \coloneqq P\indices{^\mu_\rho}U\indices{^\beta}\nabla\indices{_\beta}\left(P\indices{^\rho_\lambda}U\indices{^\alpha}\nabla
    \indices{_\alpha} V\indices{^\lambda}\right).
    \label{accel_def}
\end{equation}
Plugging in \eqref{cov_exch_V_proj}, employing
\eqref{cov_exch}, and using
\begin{equation}
    \nabla\indices{_\alpha}\nabla\indices{_\beta} C\indices{^\mu} = \nabla\indices{_\beta}\nabla\indices{_\alpha} C\indices{^\mu} + R\indices{^\mu_{\lambda\alpha\beta}}
    C\indices{^\lambda} + 2S\indices{^\lambda_{\alpha\beta}}
    \nabla\indices{_\lambda}C\indices{^\mu}
    \label{cov_deriv_comm}
\end{equation}
for any vector $C\indices{^\mu}$,
we finally find the deviation equation
for timelike curves
\begin{align*}
    a\indices{^\mu} = &R\indices{^\mu_{\lambda\alpha\beta}}U\indices{^\lambda}
    U\indices{^\alpha}V\indices{^\beta} + 
    4 S\indices{^\lambda_{\beta\alpha}}U\indices{_\lambda}
    U\indices{^\alpha}\tilde{V}\indices{^\beta}
    A\indices{^\mu}\\ &+
    A\indices{^\mu}A\indices{_\alpha}V\indices{^\alpha}
    + V\indices{^\alpha}P\indices{^\mu_\rho}\nabla\indices{_\alpha}A\indices{^\rho}\\
    &+2P\indices{^\mu_\lambda}\Bigl[U\indices{^\alpha}
    \tilde{V}\indices{^\rho}
    U\indices{^\beta}\nabla\indices{_\beta}
    S\indices{^\lambda_{\rho\alpha}}
    + U\indices{^\alpha}S\indices{^\lambda_{\rho\alpha}}
    \tilde{V}\indices{^\beta}
    \nabla\indices{_\beta}U\indices{^\rho} \\
    &+ 2U\indices{^\alpha}U\indices{^\beta}
    \tilde{V}\indices{^\xi}S\indices{^\lambda_{\rho\alpha}}
    S\indices{^\rho_{\xi\beta}}
    + S\indices{^\lambda_{\rho\alpha}}\tilde{V}\indices{^\rho}
    A\indices{^\alpha}\Bigr]. \numberthis
    \label{dev_eq}
\end{align*}

In the case of a vanishing torsion and autoparallel curves ($A\indices{^\mu} = 0$), we recover the known geodesic deviation
equation $a\indices{^\mu} = \mathring{R}\indices{^\mu_{\lambda\alpha\beta}}U\indices{^\lambda}
U\indices{^\alpha}V\indices{^\beta}$ from standard GR.

\section{Raychaudhuri Equation with Torsion}
\label{Raychaudh_sect}
We can make use of \eqref{dev_eq} to get an evolution equation
for $\nabla\indices{_\nu} U\indices{^\mu}$. This will naturally lead
us to the Raychaudhuri equation.

Compare \eqref{dev_eq} with an intermediate result of the deviation equation where the definition of the Riemann-Cartan tensor has not been employed yet, and solve for the covariant derivative of 
$\nabla\indices{_\nu} U\indices{^\mu}$ along the curves.
One finds
\begin{align*}
    P\indices{^\mu_\lambda} V\indices{^\alpha} U\indices{^\beta}
    \nabla\indices{_\beta}\left(\nabla\indices{_\alpha}
    U\indices{^\lambda}\right) = &- \tilde{V}\indices{^\beta}
    \left(\nabla\indices{_\beta}U\indices{^\alpha}\right)\left(
    \nabla\indices{_\alpha}U\indices{^\mu}\right)\\
    &- 2 U\indices{^\beta}\tilde{V}\indices{^\lambda}S\indices{^\alpha_{
    \lambda\beta}}\nabla\indices{_\alpha}U\indices{^\mu}\\
    &+ R\indices{^\mu_{\alpha\beta\lambda}}U\indices{^\alpha}
    U\indices{^\beta}V\indices{^\lambda}
    + P\indices{^\mu_\rho}V\indices{^\alpha}\nabla\indices{_\alpha}A\indices{^\rho}\\
    &- U\indices{_\beta}\tilde{V}\indices{^\beta}
    A\indices{^\alpha}\nabla\indices{_\alpha}U\indices{^\mu}.
    \numberthis
\end{align*}
After using \eqref{proj_id} and realising that there is an overall
arbitrary contraction with $\tilde{V}$, we arrive, after some manipulations, at 
the evolution equation for $\nabla\indices{_\nu} U\indices{^\mu}$
\begin{align*}
    \frac{\mathrm{D}}{\mathrm{D}\tau}\left(\nabla\indices{_\nu} U\indices{^\mu}\right) = &-\left(\nabla\indices{_\nu} U\indices{^\alpha}\right)
    \left(\nabla\indices{_\alpha} U\indices{^\mu}\right)
    - 2U\indices{^\alpha}S\indices{^\beta_{\nu\alpha}}
    \nabla\indices{_\beta} U\indices{^\mu}\\
    &+ R\indices{^\mu_{\alpha\beta\nu}}U\indices{^\alpha}
    U\indices{^\beta} + \nabla\indices{_\nu}A\indices{^\mu},
    \numberthis
    \label{evol_eq}
\end{align*}
where $\mathrm{D}/\mathrm{D}\tau \coloneqq U\indices{^\beta}
\nabla\indices{_\beta}$. We may further express 
\begin{equation}
    \nabla\indices{_\nu} U\indices{_\mu} \equiv \mathrm{D}\indices{_\nu}U\indices{_\mu}
    + \left(\delta\indices{_\nu^\alpha}\delta\indices{_\mu^\beta} -
    P\indices{_\nu^\alpha}P\indices{_\mu^\beta}\right)
    \nabla\indices{_\alpha} U\indices{_\beta},
\end{equation}
where we defined 
\begin{equation}
    \mathrm{D}\indices{_\nu}C\indices{_\mu} \coloneqq
P\indices{_\nu^\alpha}P\indices{_\mu^\beta}
\nabla\indices{_\alpha} C\indices{_\beta}
\end{equation}
for an arbitrary covector $C\indices{_\mu}$. It turns out that
$\left(\delta\indices{_\nu^\alpha}\delta\indices{_\mu^\beta} -
P\indices{_\nu^\alpha}P\indices{_\mu^\beta}\right)
\nabla\indices{_\alpha} U\indices{_\beta} = -U\indices{_\nu}
A\indices{_\mu}$ and thus 
\begin{equation}
    \nabla\indices{_\nu} U\indices{_\mu} = \mathrm{D}\indices{_\nu}U\indices{_\mu}
    -U\indices{_\nu}A\indices{_\mu}.
\end{equation}
We choose to split 
\begin{align*}
    \mathrm{D}\indices{_\nu}U\indices{_\mu} &= 
    \mathrm{D}\indices{_{(\nu}}U\indices{_{\mu)}} +
    \mathrm{D}\indices{_{[\nu}}U\indices{_{\mu]}}\\
    &= \mathrm{D}\indices{_{(\nu}}U\indices{_{\mu)}} -\frac{1}{n-1}\theta P\indices{_{\mu\nu}} + \frac{1}{n-1}\theta P\indices{_{\mu\nu}} +
    \mathrm{D}\indices{_{[\nu}}U\indices{_{\mu]}},\numberthis
\end{align*}
where 
\begin{equation}
    \theta \coloneqq \mathrm{D}\indices{^\mu}U\indices{_\mu} = \nabla\indices{^\mu}U\indices{_\mu}
\end{equation} 
is the \textit{volume scalar}/\textit{scalar expansion} describing the mean separation between curves 
and $n$ is the dimension of our manifold. Furthermore, we define
the symmetric and trace-free \textit{shear}
\begin{equation}
    \sigma\indices{_{\mu\nu}} \coloneqq \mathrm{D}\indices{_{(\nu}}U\indices{_{\mu)}} -\frac{1}{n-1}\theta P\indices{_{\mu\nu}},
\end{equation}
which measures any anisotropic evolution of curves, and the 
antisymmetric \textit{vorticity}
\begin{equation}
    \omega\indices{_{\mu\nu}} \coloneqq \mathrm{D}\indices{_{[\nu}}U\indices{_{\mu]}},
\end{equation}
describing rotational behaviour. Thereby we obtain the following split:
\begin{equation}
    \nabla\indices{_\nu} U\indices{_\mu} = \frac{1}{n-1}\theta P\indices{_{\mu\nu}} + \sigma\indices{_{\mu\nu}} +
    \omega\indices{_{\mu\nu}}
    -U\indices{_\nu}A\indices{_\mu},
    \label{nab_U_split}
\end{equation}
which further implies that 
\begin{align*}
    \nabla\indices{_\nu} A\indices{_\mu} &=
    \mathrm{D}\indices{_{\nu}}A\indices{_{\mu}} + \frac{1}{n-1}\theta
    U\indices{_\mu}A\indices{_\nu} + U\indices{_\mu}\left(
    \sigma\indices{_{\nu\alpha}} - \omega\indices{_{\nu\alpha}}\right)
    A\indices{^\alpha} \\
    &\hphantom{=}- U\indices{^\alpha}\nabla\indices{_\alpha}\left(
    U\indices{_\nu}A\indices{_\mu}\right) + A\indices{_\nu}
    A\indices{_\mu}. \numberthis
\end{align*}
Therefore we may reframe the evolution equation \eqref{evol_eq} as
\begin{align*}
    \frac{\mathrm{D}}{\mathrm{D}\tau}\left(\nabla\indices{_\nu} U\indices{_\mu}\right) = &- \frac{1}{(n-1)^2}\theta^2P\indices{_{\mu\nu}} + R\indices{_\mu_{\alpha\beta\nu}}U\indices{^\alpha}
    U\indices{^\beta}\\ &- \frac{2}{n-1}\theta\left(\sigma\indices{_{\mu\nu}} + \omega\indices{_{\mu\nu}}\right) - \sigma\indices{_{\alpha\mu}}
    \sigma\indices{_\nu^\alpha} - \omega\indices{_{\alpha\mu}}
    \omega\indices{_\nu^\alpha}\\
    &+ 2\sigma\indices{_{\alpha[\mu}}\omega\indices{_{\nu]}^\alpha}
    + \mathrm{D}\indices{_\nu}A\indices{_\mu} + 
    \frac{2}{n-1}\theta U\indices{_{(\mu}}A\indices{_{\nu)}}\\
    &+ 2 U\indices{_{(\mu}}\sigma\indices{_{\nu)\alpha}}
    A\indices{^\alpha} - 2U\indices{_{[\mu}}\omega\indices{_{\nu]\alpha}}
    A\indices{^\alpha} + A\indices{_\mu}A\indices{_\nu}\\
    &- U\indices{^\alpha}\nabla\indices{_\alpha}
    \left(A\indices{_\mu}U\indices{_\nu}\right) - \frac{2}{n-1}
    \theta S\indices{_{\mu\nu\alpha}}U\indices{^\alpha}\\
    &+ 2\left(\frac{1}{n-1}\theta U\indices{_\mu} - A\indices{_\mu}\right)U\indices{^\alpha}U\indices{^\beta}
    S\indices{_{\alpha\beta\nu}}\\
    &+ 2\left(\sigma\indices{_\mu^\beta} + \omega\indices{_\mu^\beta}\right)
    U\indices{^\alpha}S\indices{_{\beta\alpha\nu}}.
    \numberthis
    \label{evol_eq_split}
\end{align*}
Equation \eqref{evol_eq_split} encodes the full information
of the evolution of $\nabla\indices{_\nu} U\indices{_\mu}$
along the flow of the curves. We can now choose to split 
this information into several parts based on
\eqref{nab_U_split}. This provides us with three evolution
equations for each of the variables $\theta$, $\omega\indices{_{\mu\nu}}$ and $\sigma\indices{_{\mu\nu}}$.

The evolution equation for the scalar expansion $\theta$ is obtained after taking the trace of \eqref{evol_eq_split}. This is the well-known Raychaudhuri
equation in $n$ dimensions with torsion for general timelike curves
\begin{align*}
    \frac{\mathrm{d}}{\mathrm{d}\tau}\theta = &-\frac{1}{n-1}\theta^2
    -R\indices{_{(\mu\nu)}}U\indices{^\mu}U\indices{^\nu}
    -\sigma\indices{_{\mu\nu}}\sigma\indices{^{\mu\nu}}
    +\omega\indices{_{\mu\nu}}\omega\indices{^{\mu\nu}}\\
    &+ \mathrm{D}\indices{_{\mu}}A\indices{^\mu} + 
    A\indices{_\mu}A\indices{^\mu} + \frac{2}{n-1}\theta 
    S\indices{_\mu}U\indices{^\mu} - 2 S\indices{_{\mu\nu\rho}}
    U\indices{^\mu}U\indices{^\nu}A\indices{^\rho}\\
    &- 2 S\indices{_{\mu\nu\rho}}
    \sigma\indices{^{\mu\nu}}U\indices{^\rho} +
    2 S\indices{_{\mu\nu\rho}}
    \omega\indices{^{\mu\nu}}U\indices{^\rho}, \numberthis
    \label{Raychaudhuri_eq}
\end{align*}
where $S\indices{_\mu} \coloneqq S\indices{^\alpha_{\mu\alpha}}$ and $\mathrm{d}/\mathrm{d}\tau \coloneqq U\indices{^\beta}\partial\indices{_\beta}$,
in accordance with \cite{Pasmatsiou:2016bfv, Iosifidis:2018diy}. Note that the covariant derivative 
$\mathrm{D}/\mathrm{D}\tau$
reduced to the ordinary derivative $\mathrm{d}/\mathrm{d}\tau$
since their action on a scalar field, here $\theta$, is 
equivalent.

The evolution equation for the vorticity is found by antisymmetrising
\eqref{evol_eq_split}:
\begin{align*}
    \frac{\mathrm{D}}{\mathrm{D}\tau}\omega\indices{_{\mu\nu}} = 
    & R\indices{_{[\mu|\alpha\beta|\nu]}}U\indices{^\alpha}
    U\indices{^\beta} - \frac{2}{n-1}\theta \omega\indices{_{\mu\nu}}
    + 2\sigma\indices{_{\alpha[\mu}}\omega\indices{_{\nu]}^\alpha}\\
    &+ \mathrm{D}\indices{_{[\nu}}A\indices{_{\mu]}}
    - 2U\indices{_{[\mu}}\omega\indices{_{\nu]\alpha}}
    A\indices{^\alpha} - \frac{2}{n-1}
    \theta S\indices{_{[\mu\nu]\alpha}}U\indices{^\alpha}\\
    &+ 2\left(\frac{1}{n-1}\theta U\indices{_{[\mu}} - A\indices{_{[\mu}}\right)U\indices{^\alpha}U\indices{^\beta}
    S\indices{_{|\alpha\beta|\nu]}}\\
    &+ 2\left(\sigma\indices{_{[\mu}^\beta} + \omega\indices{_{[\mu}^\beta}\right)
    U\indices{^\alpha}S\indices{_{|\beta\alpha|\nu]}}.
    \numberthis 
    \label{vorticity_eq}
\end{align*}

We define the symmetric trace-free operation ``$\langle \hphantom{.}\rangle$'' on any index pair
of an arbitrary tensor $C$ of rank greater or equal 2 as
\begin{equation}
    C\indices{_{\langle\mu\nu\rangle}} \coloneqq C\indices{_{(\mu\nu)}} - 
    \frac{1}{n-1}P\indices{_{\mu\nu}}C\indices{^\alpha_\alpha}.
\end{equation}
Then 
\begin{equation}
    \nabla\indices{_{\langle\nu}}U\indices{_{\mu\rangle}} = 
    \sigma\indices{_{\mu\nu}} - U\indices{_{\langle\nu}}
    A\indices{_{\mu\rangle}}.
\end{equation}
Acting with this operation on \eqref{evol_eq_split} results
in the evolution equation for the shear
\begin{align*}
    \frac{\mathrm{D}}{\mathrm{D}\tau}\sigma\indices{_{\mu\nu}} = 
    & R\indices{_{\langle\mu|\alpha\beta|\nu\rangle}}U\indices{^\alpha}
    U\indices{^\beta} - \frac{2}{n-1}\theta \sigma\indices{_{\mu\nu}}
    - \sigma\indices{_{\alpha\langle\mu}}
    \sigma\indices{_{\nu\rangle}^\alpha}\\
    &- \omega\indices{_{\alpha\langle\mu}}
    \omega\indices{_{\nu\rangle}^\alpha} + \mathrm{D}\indices{_{\langle\nu}}A\indices{_{\mu\rangle}}
    + \frac{2}{n-1}\theta U\indices{_{(\mu}}A\indices{_{\nu)}}\\
    &+ 2 U\indices{_{(\mu}}\sigma\indices{_{\nu)\alpha}}
    A\indices{^\alpha} + A\indices{_{\langle\mu}}A\indices{_{\nu\rangle}}
    - \frac{2}{n-1}
    \theta S\indices{_{\langle\mu\nu\rangle\alpha}}U\indices{^\alpha}\\
    &+ 2\left(\frac{1}{n-1}\theta U\indices{_{\langle\mu}} - A\indices{_{\langle\mu}}\right)U\indices{^\alpha}U\indices{^\beta}
    S\indices{_{|\alpha\beta|\nu\rangle}}\\
    &+ 2\left(\sigma\indices{_{\langle\mu}^\beta} + \omega\indices{_{\langle\mu}^\beta}\right)
    U\indices{^\alpha}S\indices{_{|\beta\alpha|\nu\rangle}}.
    \numberthis 
    \label{shear_eq}
\end{align*}

\section{Geometric Kinematic Variables}
\label{Geom_var_sect}
The previous sections are based on a particular definition
of the kinematic variables $\theta$, $\sigma\indices{_{\mu\nu}}$, and $\omega\indices{_{\mu\nu}}$
in terms of a decomposition of $\nabla\indices{_\nu}
U\indices{_\mu}$. This choice is used in a number of references such as \cite{Pasmatsiou:2016bfv, Iosifidis:2018diy} and roots in the direct analogy to
the torsionless case. However, as evident from \eqref{cov_exch_V}, the complete change of the separation
vector along the curves is not fully determined by
$\nabla\indices{_\nu}U\indices{_\mu}$ but there are 
further components such as torsion and the autoparallel acceleration
$A\indices{^\mu}$. Hence there is a point to be made for 
decomposing not just $\nabla\indices{_\nu}U\indices{_\mu}$
but the whole r.h.s. of \eqref{cov_exch_V} into its 
trace, symmetric tracefree and antisymmetric part.
This path is taken in, e.g., \cite{Luz:2017ldh,Luz:2019frs,Luz:2019kmm}.
Here we shall introduce these so-called geometric kinematic variables,
henceforth labelled by a ``(g)'' superscript,
and relate them to the ones used before.

We can express \eqref{cov_exch_V} as
\begin{equation}
    U\indices{^\beta}\nabla
    \indices{_\beta} V\indices{^\mu} =
    V\indices{^\beta}\nabla
    \indices{_\beta} U\indices{^\mu}
    + U\indices{^\mu}A\indices{_\beta}V\indices{^\beta}
    + 2 U\indices{^\alpha}V\indices{^\beta} 
    P\indices{^\mu_\nu}
    S\indices{^\nu_{
    \beta\alpha}},
\end{equation}
or, after projecting, as
\begin{equation}
    P\indices{^\mu_\lambda}U\indices{^\beta}\nabla
    \indices{_\beta} V\indices{^\lambda} = B\indices{_\beta^\mu}
    V\indices{^\beta}
\end{equation}
with
\begin{align*}
    B\indices{_\beta^\mu} &\coloneqq
    P\indices{^\mu_\lambda}P\indices{_\beta^\rho}\Bigl(
    \nabla\indices{_\rho} U\indices{^\lambda}
    + U\indices{^\lambda}A\indices{_\rho}    
    + 2 U\indices{^\alpha}
    P\indices{^\lambda_\nu}S\indices{^\nu_{\rho\alpha}}
    \Bigr)\\
    &\hphantom{:}= \mathrm{D}\indices{_\beta}U\indices{^\mu}
    + 2U\indices{^\alpha}P\indices{^\mu_\nu}
    S\indices{^\nu_{\beta\alpha}}.
    \numberthis
\end{align*}
Now instead of just decomposing $\mathrm{D}\indices{_\beta}U\indices{^\mu}$, we decompose the 
whole tensor $B\indices{_\beta^\mu}$ as
\begin{equation}
    B\indices{_{\nu\mu}} = \frac{1}{n-1}\theta^{(g)}
    P\indices{_{\mu\nu}} + {\sigma^{(g)}}\indices{_{\mu\nu}}
    + {\omega^{(g)}}\indices{_{\mu\nu}},
\end{equation}
where
\begin{align*}
    \theta^{(g)} &\coloneqq B\indices{_\alpha^\alpha}\\
    {\sigma^{(g)}}\indices{_{\mu\nu}} &\coloneqq
    B\indices{_{(\nu\mu)}} - \frac{1}{n-1}\theta^{(g)}
    P\indices{_{\mu\nu}}\\
    {\omega^{(g)}}\indices{_{\mu\nu}} &\coloneqq
    B\indices{_{[\nu\mu]}}.\numberthis
\end{align*}
These are related to the previously introduced variables by direct computation
\begin{align*}
    \theta^{(g)} &= \theta - 2U\indices{^\alpha}S\indices{_\alpha}\\
    {\sigma^{(g)}}\indices{_{\mu\nu}} &= \sigma\indices{_{\mu\nu}} + 2U\indices{^\alpha}P\indices{_{\langle\mu|\beta}}
    S\indices{^\beta_{|\nu\rangle\alpha}}\\
    {\omega^{(g)}}\indices{_{\mu\nu}} &=
    \omega\indices{_{\mu\nu}} +
    2U\indices{^\alpha}P\indices{_{[\mu|\beta}}
    S\indices{^\beta_{|\nu]\alpha}}.
    \numberthis
    \label{geom_vs_old}
\end{align*}
Writing the evolution equations in terms of the geometric
variables or in the old variables is a matter of taste.
The torsional terms which are absorbed in the geometric
variables in \eqref{geom_vs_old} appear explicitly on the 
r.h.s. of the evolution equations in the previous section,
while in terms of geometric variables only appear implicitly
within them. The evolution equations in the geometric variables therefore look different. However they include 
exactly the same physical information.

Nevertheless, our later investigations of the Raychaudhuri equation must be conducted in terms of the geometric variables. The explicit form of the equation is crucial for establishing appropriate ECs and ultimately determining the formation of focal points.
Therefore, in the following we shall proceed to work
with the Raychaudhuri equation in geometric variables
\begin{align*}
    \frac{\mathrm{d}}{\mathrm{d}\tau}\theta^{(g)} = &-\frac{1}{n-1}{\theta^{(g)}}^2
    -R\indices{_{(\mu\nu)}}U\indices{^\mu}U\indices{^\nu}
    -{\sigma^{(g)}}\indices{_{\mu\nu}}{\sigma^{(g)}}\indices{^{\mu\nu}}\\
    &+{\omega^{(g)}}\indices{_{\mu\nu}}{\omega^{(g)}}\indices{^{\mu\nu}}
    + \mathrm{D}\indices{_{\mu}}A\indices{^\mu} + 
    A\indices{_\mu}A\indices{^\mu}\\
    &- \frac{2}{n-1}\theta^{(g)} 
    S\indices{_\mu}U\indices{^\mu}
    - 2 S\indices{_{\mu\nu\rho}}
    U\indices{^\mu}U\indices{^\nu}A\indices{^\rho}\\
    &+ 2 S\indices{_{\mu\nu\rho}}
    {\sigma^{(g)}}\indices{^{\mu\nu}}U\indices{^\rho}
    -
    2 S\indices{_{\mu\nu\rho}}
    {\omega^{(g)}}\indices{^{\mu\nu}}U\indices{^\rho}\\
    &-2\frac{\mathrm{D}}{\mathrm{D}\tau}\left(
    U\indices{^\alpha}S\indices{_\alpha}\right), \numberthis
    \label{Geom_Raychaudhuri_eq}
\end{align*}
which is in agreement with \cite{Luz:2017ldh}. Note that different index placements in the definitions of the covariant derivative and the vorticity tensor result in the fact that our torsion and vorticity tensors admit a relative minus sign compared to the torsion and vorticity defined in \cite{Luz:2017ldh}.

\section{Conjugate Points}
\label{Conj_points_sec}
The notion of conjugate points, to be defined below,
is a vital component in establishing and proving the singularity theorem later on. To discuss them, we have to 
revisit the deviation equation \eqref{dev_eq}:
\begin{align*}      
    a\indices{^\mu} = &R\indices{^\mu_{\lambda\alpha\beta}}U\indices{^\lambda}
    U\indices{^\alpha}V\indices{^\beta}
    + 4S\indices{^\lambda_{\beta\alpha}}U\indices{_\lambda}
    U\indices{^\alpha}V\indices{^\beta}A\indices{^\mu}\\
    &+ A\indices{^\mu}A\indices{_\beta}V\indices{^\beta}
    +P\indices{^\mu_\lambda}V\indices{^\beta}\nabla\indices{_\beta}A\indices{^\lambda}\\
    &+  
    2P\indices{^\mu_\lambda}\Bigl[U\indices{^\alpha}
    V\indices{^\beta}U\indices{^\rho}\nabla\indices{_\rho}
    S\indices{^\lambda_{\beta\alpha}}
    + U\indices{^\alpha}S\indices{^\lambda_{\rho\alpha}}
    V\indices{^\beta}
    \nabla\indices{_\beta}U\indices{^\rho}\\
    &+ 2U\indices{^\alpha}U\indices{^\xi}
    V\indices{^\beta}S\indices{^\lambda_{\rho\alpha}}
    S\indices{^\rho_{\beta\xi}} + 
    S\indices{^\lambda_{\beta\alpha}}
    V\indices{^\beta}A\indices{^\alpha}\Bigr]. \numberthis
    \label{Jacobi_eq_1}
\end{align*}
Notice that here we expressed \eqref{dev_eq} solely
in terms of $V\indices{^\beta}$, instead of the 
unprojected $\tilde{V}\indices{^\beta}$. The reason
for this will become clear soon.
We further have the definition \eqref{accel_def}, which may be reframed as
\begin{align*}
    a\indices{^\mu} = &\frac{\mathrm{D}^2}{\mathrm{D}\tau^2}V\indices{^\mu} + 2U\indices{^\mu}
    U\indices{^\lambda}V\indices{^\beta}\nabla\indices{_\beta}
    A\indices{_\lambda}
    -4U\indices{^\mu}U\indices{^\alpha}A\indices{_\lambda}
    S\indices{^\lambda_{\beta\alpha}}V\indices{^\beta}\\
    &- U\indices{^\mu}V\indices{^\beta}U\indices{^\alpha}
    \nabla\indices{_\alpha}A\indices{_\beta}
    -A\indices{^\mu}V\indices{^\beta}A\indices{_\beta}.
\end{align*}
Together with \eqref{Jacobi_eq_1} we thus find
\begin{align*}
    \frac{\mathrm{D}^2}{\mathrm{D}\tau^2}V\indices{^\mu} = 
    &\Bigl[R\indices{^\mu_{\lambda\alpha\beta}}U\indices{^\lambda}
    U\indices{^\alpha}
    + 4S\indices{^\lambda_{\beta\alpha}}U\indices{^\alpha}(U\indices{_\lambda}
    A\indices{^\mu} + A\indices{_\lambda}U\indices{^\mu})\\
    &+ 2A\indices{^\mu}A\indices{_\beta}
    + (\nabla\indices{_\beta}A\indices{^\mu})\\
    &+ U\indices{^\mu}(U\indices{^\alpha}\{\nabla\indices{_\alpha}A\indices{_\beta}\} - U\indices{^\alpha}\{\nabla\indices{_\beta}A\indices{_\alpha}\})
    \\
    &+  
    2P\indices{^\mu_\lambda}\Bigl(U\indices{^\alpha}
    U\indices{^\rho}\{\nabla\indices{_\rho}
    S\indices{^\lambda_{\beta\alpha}}\}
    + U\indices{^\alpha}S\indices{^\lambda_{\rho\alpha}}
    \{\nabla\indices{_\beta}U\indices{^\rho}\}\\
    &+ 2U\indices{^\alpha}U\indices{^\xi}
    S\indices{^\lambda_{\rho\alpha}}
    S\indices{^\rho_{\beta\xi}}
    + S\indices{^\lambda_{\beta\alpha}}A\indices{^\alpha}
    \Bigr)\Bigr]
    V\indices{^\beta}. \numberthis
    \label{Jacobi_eq_2}
\end{align*}
By writing out the covariant derivatives on the l.h.s. explicitly, we may express this deviation equation
as an ordinary, linear, second order differential equation\footnote{
   It is possible to choose a tetrad frame, that is parallely transported along our curve of interest through
    the Fermi derivative,
    in which the spin-connection components vanish \cite{Aldrovandi:2002zj}. 
    In this case, the connection reduces to the Weitzenböck connection
    $\Gamma\indices{^\mu_{\alpha\beta}} = e\indices{_c^\mu}\partial\indices{_\beta}e\indices{^c_\alpha}$.
    On a mathematical level, this fact 
    ensures the validity of the strong Equivalence Principle in presence of torsion. Compared to the standard GR case, this 
    does not simplify our equation substantially, hence we choose
    to continue with the general connection.}
\begin{align*}
    \frac{\mathrm{d}^2}{\mathrm{d}\tau^2}V\indices{^\mu} = 
    &\Bigl[R\indices{^\mu_{\lambda\alpha\beta}}U\indices{^\lambda}
    U\indices{^\alpha}
    + 4S\indices{^\lambda_{\beta\alpha}}U\indices{^\alpha}(U\indices{_\lambda}
    A\indices{^\mu} + A\indices{_\lambda}U\indices{^\mu})\\
    &+ 2A\indices{^\mu}A\indices{_\beta}
    + (\nabla\indices{_\beta}A\indices{^\mu})\\
    &+ U\indices{^\mu}(U\indices{^\alpha}\{\nabla\indices{_\alpha}A\indices{_\beta}\} - U\indices{^\alpha}\{\nabla\indices{_\beta}A\indices{_\alpha}\})
    \\
    &+  
    2P\indices{^\mu_\lambda}\Bigl(U\indices{^\alpha}
    U\indices{^\rho}\{\nabla\indices{_\rho}
    S\indices{^\lambda_{\beta\alpha}}\}
    + U\indices{^\alpha}S\indices{^\lambda_{\rho\alpha}}
    \{\nabla\indices{_\beta}U\indices{^\rho}\}\\
    &+ 2U\indices{^\alpha}U\indices{^\xi}
    S\indices{^\lambda_{\rho\alpha}}
    S\indices{^\rho_{\beta\xi}}
    + S\indices{^\lambda_{\beta\alpha}}A\indices{^\alpha}
    \Bigr)\\
    &-\frac{\mathrm{d}}{\mathrm{d}\tau}(
    U\indices{^\gamma}\Gamma\indices{^\mu_{\beta\gamma}})
    + U\indices{^\alpha}U\indices{^\gamma}
    \Gamma\indices{^\mu_{\lambda\alpha}}
    \Gamma\indices{^\lambda_{\beta\gamma}}\\
    &-2U\indices{^\gamma}\Gamma\indices{^\mu_{\lambda\gamma}}
    \Bigl(U\indices{^\lambda}A\indices{_\beta} + 
    \{\nabla\indices{_\beta}U\indices{^\lambda}\}\\
    &+2U\indices{^\alpha}P\indices{^\mu_\nu}
    S\indices{^\nu_{\beta\alpha}}
    \Bigr)
    \Bigr]
    V\indices{^\beta}. \numberthis
    \label{Jacobi_eq_3}
\end{align*}

Generalising the 
definition of a Jacobi field \cite{wald:1984} to the case
including torsion, we call a solution $V\indices{^\mu}$ of the
deviation equation \eqref{Jacobi_eq_3}
a \textit{Jacobi field} on $\gamma_{s'}(\tau)$, where $s'$ has some fixed but arbitrary value.
As in the standard case, we call two points $p,q \in \gamma_{s'}(\tau)$ \textit{conjugate} if there exists a Jacobi field that
is not identically zero and vanishes at $p$ and $q$, see Fig.\,\ref{fig.5.1}.

\begin{figure}[h]
	\begin{tikzpicture}
		\draw[thick] (4,-0.5) .. controls (5,2) .. (7,3);
        \draw[thick] (4,-0.5) .. controls (6,1) .. (7,3);
        \filldraw[black] (4,-0.5) circle (2pt) node[xshift=0.2cm,yshift=-0.2cm]{$p$};
        \filldraw[black] (7,3) circle (2pt) node[xshift=0.2cm,yshift=-0.2cm]{$q$};
        \draw[thick,->,color = blue] (5.53,0.7) -- (6.4,1.4)
        node[xshift=0.3cm,yshift=-0.2cm]{$U\indices{^\mu}$};
        \draw[thick,->,color = red] (5.53,0.7) -- (4.87,1.5)
        node[xshift=0.1cm,yshift=-0.7cm]{$V\indices{^\mu}$};
	\end{tikzpicture}
	\caption{Depiction of two conjugate points $p$ and $q$ along the flow of the congruence}
	\label{fig.5.1}
\end{figure}
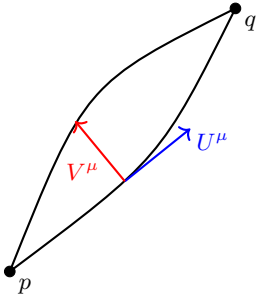

Choose the parametrisation of the curves such that $p$
lies at $\tau = 0$. Following \cite{wald:1984}, a general solution of the linear, second order ordinary
differential equation \eqref{Jacobi_eq_3} then admits the form
\begin{equation}
    V\indices{^\mu}(\tau) = E\indices{^\mu_\nu}(\tau)\frac{\mathrm{d}V\indices{^\nu}}{\mathrm{d}\tau}(0)
    \label{Ansatz:V}
\end{equation}
for some matrix $E\indices{^\mu_\nu}(\tau)$,
since $V\indices{^\mu}(0) = 0$ per assumption. Plugging this ansatz back into \eqref{Jacobi_eq_3} yields the corresponding
differential equation for the matrix $E$.

At the point $q$ we are required to have $V\indices{^\mu}(\tau_q) = 0$. Then the homogeneous linear system
\begin{equation}
    0 = E\indices{^\mu_\nu}(\tau_q)\frac{\mathrm{d}V\indices{^\nu}}{\mathrm{d}\tau}(0)
\end{equation}
has a non-trivial solution only if $\mathrm{det}(E) = 0$ at $q$.
From \eqref{cov_exch_V} it follows that
\begin{equation}
    \frac{\mathrm{D}}{\mathrm{D}\tau}V\indices{^\mu}
    =\left[U\indices{^\mu}A\indices{_\beta} + \nabla\indices{_\beta}U\indices{^\mu} + 
    2U\indices{^\alpha}P\indices{^\mu_\nu}S\indices{^\nu_{\beta\alpha}}\right]V\indices{^\beta}
\end{equation}
and with \eqref{Ansatz:V} therefore
\begin{align*}
    \frac{\mathrm{d}}{\mathrm{d}\tau}E\indices{^\mu_\nu}
    =&\Bigl[U\indices{^\mu}A\indices{_\beta} + \nabla\indices{_\beta}U\indices{^\mu} + 
    2U\indices{^\alpha}P\indices{^\mu_\nu}S\indices{^\nu_{\beta\alpha}}\\
    &-U\indices{^\alpha}\Gamma\indices{^\mu_{\beta\alpha}}\Bigr]E\indices{^\beta_\nu}.\numberthis
\end{align*}
In matrix notation this reads
\begin{equation}
    C = \left( \frac{\mathrm{d}}{\mathrm{d}\tau}E\right)E^{-1},
    \label{matrix_notat_C}
\end{equation}
where
\begin{equation}
    C\indices{^\mu_\beta}\coloneqq 
    U\indices{^\mu}A\indices{_\beta} + \nabla\indices{_\beta}U\indices{^\mu} + 
    2U\indices{^\alpha}P\indices{^\mu_\nu}S\indices{^\nu_{\beta\alpha}}
    -U\indices{^\alpha}\Gamma\indices{^\mu_{\beta\alpha}}.
    \label{Def:C}
\end{equation}
On the one hand taking the trace of \eqref{matrix_notat_C} yields
\begin{equation}
    \mathrm{tr}(C) = \mathrm{tr}\left(\left( \frac{\mathrm{d}}{\mathrm{d}\tau}E\right)E^{-1}\right) = 
    \frac{1}{\mathrm{det}(E)}\frac{\mathrm{d}}{\mathrm{d}\tau}
    \mathrm{det}(E),
\end{equation}
whereas from \eqref{Def:C} we get
\begin{equation}
    \mathrm{tr}(C) = \theta - U\indices{^\alpha}[2S\indices{_\alpha}+ \Gamma\indices{^\mu_{\mu\alpha}}].
\end{equation}
With $\theta^{(g)} = \theta - 2U\indices{^\alpha}S\indices{_\alpha}$ and 
$\Gamma\indices{^\mu_{\mu\alpha}} = \mathring{\Gamma}\indices{^\mu_{\mu\alpha}}$
we conclude
\begin{equation}
    \theta^{(g)} = \frac{\mathrm{d}}{\mathrm{d}\tau}\mathrm{ln}\abs{\mathrm{det}(E)} + 
    U\indices{^\alpha}
    \mathring{\Gamma}\indices{^\mu_{\mu\alpha}}.
    \label{Theta_conj}
\end{equation}

As the timelike curves approach the point $q$, $\mathrm{det}(E)$ tends to $0$. Thus
$\frac{\mathrm{d}}{\mathrm{d}\tau}\mathrm{ln}\abs{\mathrm{det}(E)} \to -\infty$ and hence also 
$\theta^{(g)} \to -\infty$. The converse also holds since the connection is assumed to admit only finite values.
Therefore, even with torsion, we retain the fact that the 
points $p$ and $q$ are conjugate if and only if 
$\lim_{\tau\to\tau_q}\theta^{(g)} = -\infty$, as in standard GR.
However, the question where $E$ becomes singular highly depends on the
torsional contributions present in the respective differential equation.

The standard GR form of \eqref{Theta_conj} is recovered
by setting torsion to zero and by going to a local
inertial frame where the connection vanishes.

\section{Hypersurface Orthogonality}
\label{Hypersur_orthog_sec}
Our congruence of timelike curves is not completely free
but has to obey certain conditions in order to be
physically feasible. More specifically, we require the 
vector field $U$ to be hypersurface orthogonal. This means 
that at every point where $U$ is defined, there exists locally
an embedded submanifold which is orthogonal to $U$,
see Fig\,\ref{fig.6.1}.
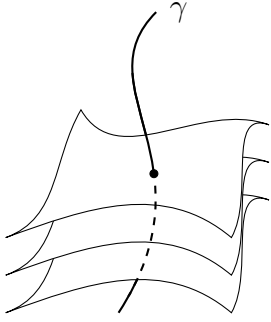
\begin{figure}[h]
	\begin{tikzpicture}
		\draw[thick] (5,-1) .. controls (6.3,1) and (4.5,2).. (5.5,3) node[xshift=0.3cm]{\large{$\gamma$}};
        \draw[fill=white,opacity=1] (3.5, -1) to[out=20, in=140] (6.5, -1) to [out=60, in=160]
        (7, 0.5) to[out=160, in=300] (4.5,0.7) to[out=240, in=20]
        cycle;
        \draw[fill=white,opacity=1] (3.5, -0.5) to[out=20, in=140] (6.5, -0.5) to [out=60, in=160]
        (7, 1) to[out=160, in=300] (4.5,1.2) to[out=240, in=20]
        cycle;
        \draw[fill=white,opacity=1] (3.5, 0) to[out=20, in=140] (6.5, 0) to [out=60, in=160]
        (7, 1.5) to[out=160, in=300] (4.5,1.7) to[out=240, in=20]
        cycle;
        \begin{scope}
            \clip(5.3,0) circle (0.8);
            \draw[thick,dashed,opacity=1] (5,-1) .. controls (6.3,1) and (4.5,2).. (5.5,3);
        \end{scope}
        \begin{scope}
            \clip(5.3,1.5) circle (0.7);
            \draw[thick,opacity = 1] (5,-1) .. controls (6.3,1) and (4.5,2).. (5.5,3);
        \end{scope}
        \filldraw[black] (5.47,0.85) circle (1.5pt);
	\end{tikzpicture}
	\caption{Illustration of a foliation of spacetime into
    hypersurfaces orthogonal to the tangent vector field of the curve $\gamma$}
	\label{fig.6.1}
\end{figure}
If this were not the case then it would be impossible to 
foliate spacetime into equal-time 3-surfaces and hence
the standard $\Lambda$CDM treatment of the Universe would
collapse.

Mathematically, hypersurface orthogonality is realised 
with Frobenius' theorem. As in \cite{wald:1984,Luz:2019kmm},
let $p\in \mathcal{M}$ and consider subspaces $W_p\subseteq T_p\mathcal{M}$ and $U_p^*\subseteq T_p^*\mathcal{M}$ with
$\mathrm{dim}(W_p) = m \geq 1$, $m < n$ and $\mathrm{dim}(U_p^*) = n-m$ such that for all $u \in U_p^*$ and all
$q \in W_p$
\begin{equation}
    u(q) = 0.
\end{equation}
This is possible since any subspace of an ambient vector space equipped with a non-degenerate symmetric bilinear form,
together with its orthogonal complement, have complementary dimension.
Here, $U_p$ is the orthogonal complement of $W_p$.
Let $O_p$ be some open neighbourhood of $p$ and define
\begin{align*}
    W &\coloneqq \bigsqcup_{q\in O_p} W_q \subseteq T\mathcal{M}\\
    U^* &\coloneqq\bigsqcup_{q\in O_p} U_q^* \subseteq T\mathcal{M}^*,
    \numberthis
\end{align*}
where $T\mathcal{M}$ is the tangent bundle and ``$\sqcup$'' denotes the
disjoint union. Frobenius' theorem states that $W$ admits
integral submanifolds if and only if for all $X,Y \in \Gamma(W)$ and for all $u\in \Gamma(U^*)$ 
\begin{equation}
    u(\mathcal{L}_X Y) = 0,
    \label{Thm:Frobenius}
\end{equation}
where $\Gamma(W)$ and $\Gamma(U^*)$ denote the set of all sections of the bundle $W$ and the cotangent bundle $U^*$ respectively. But $\mathcal{L}_X Y = [X,Y]$ and hence
\eqref{Thm:Frobenius} reads ${u_\alpha[X,Y]^\alpha = 0}$.
Writing the commutator as in \eqref{commutator1}, pulling
$u_\alpha$ into the covariant derivative and using the 
requirement that $u_\alpha Y^\alpha = 0 = u_\alpha X^\alpha$,
we arrive at
\begin{equation}
    X^\alpha Y^\beta\left(\nabla\indices{_{[\beta}}u\indices{_{\alpha]}}+ S\indices{^\gamma_{\alpha\beta}}u\indices{_\gamma}\right)
    = 0.
\end{equation}
Since $X$ and $Y$ are arbitrary, we find equivalently that
\begin{equation}
    \nabla\indices{_{[\alpha}}u\indices{_{\beta]}}-S\indices{^\gamma_{\alpha\beta}}u\indices{_\gamma}
    = \sum_{i=1}^{n-m}(u_i)\indices{_{[\alpha}}(v_i)\indices{_{\beta]}},
\end{equation}
where the $u_i$ are $n-m$ linearly independent elements 
of $\Gamma(U^*)$ (thus in particular $(u_i)_\alpha Y^\alpha = 0 = (u_i)_\alpha X^\alpha$ for all $i$) and $(v_i)$ are arbitrary
elements in $\Gamma(U^*)$.

Let us consider the special case where $\mathrm{dim}(U^*) = 1$,
which is the case relevant for us. Then we may choose $u$ itself as $(u_1)$ and we get 
\begin{equation}
    \nabla\indices{_{[\alpha}}u\indices{_{\beta]}}-S\indices{^\gamma_{\alpha\beta}}u\indices{_\gamma}
    = u\indices{_{[\alpha}}v\indices{_{\beta]}},
\end{equation}
where $v \in \Gamma(U^*)$ arbitrary. Choose $v$ such that
${P\indices{_\nu^\beta}v\indices{_\beta} = 0}$ and contract
the whole equation above with $P\indices{_\mu^\alpha}P\indices{_\nu^\beta}$.
With $\omega\indices{_{\nu\mu}} = \mathrm{D}\indices{_{[\mu}}u\indices{_{\nu]}}$ we finally get
\begin{equation}
    \omega\indices{_{\mu\nu}} = S\indices{^\alpha_{\nu\mu}}u\indices{_\alpha}
    + 2S\indices{^\alpha_{[\nu|\beta|}}u\indices{_{\mu]}}
    u\indices{_\alpha}u\indices{^\beta}.
\end{equation}
This holds for an arbitrary $u\in \Gamma(U^*)$. For our specific 
case we choose to consider the vector field $U\indices{^\alpha}$ tangent 
to the timelike curves and its associated 1-form defined via
the musical isomorphism $U\indices{_\alpha} \coloneqq g\indices{_{\alpha\beta}}U\indices{^\beta}$.
Then the condition for $U\indices{^\alpha}$ to be hypersurface
orthogonal reads
\begin{equation}
    \omega\indices{_{\mu\nu}} = S\indices{^\alpha_{\nu\mu}}U\indices{_\alpha}
    + 2S\indices{^\alpha_{[\nu|\beta|}}U\indices{_{\mu]}}
    U\indices{_\alpha}U\indices{^\beta},
\end{equation}
with $\omega\indices{_{\nu\mu}} =
\mathrm{D}\indices{_{[\mu}}U\indices{_{\nu]}}$ the vorticity
encountered previously,
or 
\begin{equation}
    {\omega^{(g)}}\indices{_{\mu\nu}} = \left(S\indices{_{\alpha\nu\mu}} + 2S\indices{_{[\mu\nu]\alpha}}\right)U\indices{^\alpha}
    + 4U\indices{_{[\mu}}S\indices{^\alpha_{\nu]\beta}}
    U\indices{_\alpha}U\indices{^\beta}
\end{equation}
in geometric variables, in agreement with \cite{Luz:2019kmm}.
For a totally antisymmetric torsion, these conditions reduce
to
\begin{equation}
    \omega\indices{_{\mu\nu}} = S\indices{^\alpha_{\nu\mu}}U\indices{_\alpha}
\end{equation}
and
\begin{equation}
    {\omega^{(g)}}\indices{_{\mu\nu}} = S\indices{_{\mu\nu}^\alpha}U\indices{_\alpha}.
\end{equation}

Therefore, in the presence of torsion, we can no longer assume that vorticity vanishes, as is typically done in the standard GR treatment. It must be emphasized that torsion, and torsion alone, is responsible for the generation of vorticity.

\section{Case studies}
\label{Case_stud_sec}
The Raychaudhuri equation allows us to predict the appearances
of conjugate points (also referred to as \textit{focal points}) provided that we assume an energy condition.
For the case of timelike curves, as studied here, the relevant
condition is the strong energy condition (SEC).

In this section we are going to investigate different specific
torsion forms and examine how a modified SEC with torsion
leads to the appearance of focal points.
Torsion may be split, with respect to the action of the Lorentz group, into three irreducible tensors \cite{Capozziello:2001mq}: A totally antisymmetric part,
a vectorial contribution and the rest which is traceless
but not totally antisymmetric.
We are going to investigate the former two parts in detail and 
consider the special case of a Weyssenhoff fluid
relevant most prominently in Einstein-Cartan-Sciama-Kibble (ECSK) theory. For a recap of the most well-known energy
conditions, see \Cref{sect_en_cond}.

\subsection{Totally Antisymmetric Torsion}
In the case of a totally antisymmetric torsion tensor
$S\indices{_{\mu\nu\lambda}} = S\indices{_{[\mu\nu\lambda]}}$
the geometric expansion scalar coincides with the conventional one $\theta^{(g)} = \theta$.
For autoparallel curves, and for $n=4$,
the Raychaudhuri equation \eqref{Raychaudhuri_eq} reads
\begin{align*}
    \frac{\mathrm{d}}{\mathrm{d}\tau}\theta = &-\frac{1}{3}\theta^2
    -R\indices{_{\mu\nu}}U\indices{^\mu}U\indices{^\nu}
    -\sigma\indices{_{\mu\nu}}\sigma\indices{^{\mu\nu}}
    +\omega\indices{_{\mu\nu}}\omega\indices{^{\mu\nu}}\\
    & + 2 S\indices{_{\mu\nu\rho}}
    \omega\indices{^{\mu\nu}}U\indices{^\rho}. \numberthis
\end{align*}
Requiring hypersurface orthogonality, this can be reframed
into
\begin{align*}
    \frac{\mathrm{d}}{\mathrm{d}\tau}\theta = &-\frac{1}{3}\theta^2
    -R\indices{_{\mu\nu}}U\indices{^\mu}U\indices{^\nu}
    -\sigma\indices{_{\mu\nu}}\sigma\indices{^{\mu\nu}}\\
    &-S\indices{^\alpha_{\mu\nu}}S\indices{^\beta^{\mu\nu}}U\indices{_\alpha}U\indices{_\beta}. \numberthis
\end{align*}
One can show that in the case of a totally antisymmetric
torsion tensor,
\begin{equation}
    R\indices{_{\mu\nu}}U\indices{^\mu}U\indices{^\nu} = 
    \mathring{R}\indices{_{\mu\nu}}
    U\indices{^\mu}U\indices{^\nu} - 
    S\indices{^\alpha_{\mu\nu}}S\indices{^\beta^{\mu\nu}}
    U\indices{_\alpha}U\indices{_\beta}.
\end{equation}
Therefore we conclude that in this case, the Raychaudhuri equation does
not contain any torsional contributions and reduces to its standard form
\begin{equation}
    \frac{\mathrm{d}}{\mathrm{d}\tau}\theta = -\frac{1}{3}\theta^2
    -\mathring{R}\indices{_{\mu\nu}}U\indices{^\mu}U\indices{^\nu}
    -\sigma\indices{_{\mu\nu}}\sigma\indices{^{\mu\nu}}. \numberthis
\end{equation}
Additionally one may verify that $\theta = \mathring{\theta}$
and $\sigma\indices{_{\mu\nu}} = \mathring{\sigma}\indices{_{\mu\nu}}$.

Since $\sigma\indices{_{\mu\nu}}$ and 
$\omega\indices{_{\mu\nu}}$ have both their indices projected 
orthogonal to $U\indices{^\alpha}$, they are both spacelike
in the sense that
$\sigma\indices{_{\mu\nu}}\sigma\indices{^{\mu\nu}}\geq 0$
and $\omega\indices{_{\mu\nu}}\omega\indices{^{\mu\nu}}\geq 0$.
Hence we can generalise Wald's proposition \cite[p.\,226]{wald:1984} to:

\begin{proposition}
Let $(\mathcal{M},g)$ be a spacetime, equipped with an independent
metric compatible affine connection $\Gamma$. Let the Ricci tensor satisfy $\mathring{R}\indices{_{\mu\nu}}\xi\indices{^\mu}\xi\indices{^\nu} \geq 0$
for all timelike $\xi\indices{^\mu}$ and let the torsion tensor $S\indices{_{\mu\nu\lambda}}$ be totally antisymmetric. Consider a hypersurface orthogonal congruence
of timelike autoparallels emanating from $p \in \mathcal{M}$.
If $\theta_0\coloneqq \theta(\tau_0) < 0$ at some 
point $r\in\mathcal{M}$ at parameter value $\tau_0$ of a member of the congruence in the future 
of $p$, then there exists a point $q$ conjugate to $p$ within
proper time $\tau\leq -3/\theta_0$ from $r$ (assuming the congruence extends that far).
\end{proposition}
\begin{proof}
The proof follows quickly based on the discussion before:
Note that 
\begin{equation}
    \frac{\mathrm{d}}{\mathrm{d}\tau}\theta \leq -\frac{1}{3}\theta^2
\end{equation}
supposed that all assumptions in the proposition hold.
Upon choosing $\tau_0 = 0$, we can solve this differential inequality to find
\begin{equation}
    \theta(\tau)^{-1} \geq \theta_0^{-1} + \frac{1}{3}\tau.
\end{equation}
Now if $\theta_0 < 0$ at some point $r$, then $\theta^{-1}$
must cross zero within the time interval $\tau\in [0,-3/\theta_0]$. Since $\theta$ is negative before the crossing, this means that $\theta \to -\infty$ within this 
time interval. But $\theta \to -\infty$ is equivalent 
to the existence of a conjugate point.
\end{proof}

Note that a ``spacetime'' is defined in the sense of Wald,
namely a smooth manifold equipped with a Lorentzian metric.
Further, due to the fact that torsion completely cancels out here,
the SEC $\mathring{R}\indices{_{\mu\nu}}\xi\indices{^\mu}\xi\indices{^\nu} \geq 0$ is not modified at all.

It is likewise not difficult to show that 
${\sigma^{(g)}}\indices{_{\mu\nu}} = \sigma\indices{_{\mu\nu}}$ and that
the shear equation
\eqref{shear_eq} also reduces to its standard form without torsional contributions:

\begin{align*}
    \frac{\mathring{\mathrm{D}}}{\mathrm{D}\tau}\sigma\indices{_{\mu\nu}} = 
    & \mathring{R}\indices{_{\langle\mu|\alpha\beta|\nu\rangle}}U\indices{^\alpha}
    U\indices{^\beta} - \frac{2}{3}\theta \sigma\indices{_{\mu\nu}}
    - \sigma\indices{_{\alpha\langle\mu}}
    \sigma\indices{_{\nu\rangle}^\alpha}.
    \numberthis 
\end{align*}
Hence there are also no residual
direction-dependent torsional effects to be expected.

Lastly, the evolution equation for the vorticity
\eqref{vorticity_eq} becomes
\begin{equation}
    U\indices{^\alpha}
    U\indices{^\beta}
    \left(\mathring{R}\indices{_{[\mu|\alpha\beta|\nu]}}
    + 2S\indices{_\beta^\gamma_{[\mu}}
    S\indices{_{|\gamma\alpha|\nu]}}\right) = 0.
\end{equation}
But $\mathring{R}\indices{_{[\mu|\alpha\beta|\nu]}}
U\indices{^\alpha}U\indices{^\beta} = 0$, due to the symmetries
of the Riemann-tensor, and $U\indices{^\alpha}U\indices{^\beta}
S\indices{_\beta^\gamma_{[\mu}}
S\indices{_{|\gamma\alpha|\nu]}} = 0$ hold separately. Therefore
the evolution equation for the vorticity is merely an identity
and does not provide us with additional constraints.

\subsection{Vectorial Torsion}
The second irreducible component of torsion is given by
(first introduced in \cite{Friedmann1924}, with applications discussed, e.g., in \cite{Csillag:2024oqo})
\begin{equation}
    S\indices{_{\mu\nu\lambda}} = -\frac{2}{3}g\indices{_{\mu[\nu}}S\indices{_{\lambda]}}.
    \label{vector_tors}
\end{equation}

Hypersurface orthogonality demands here that
${\omega^{(g)}}\indices{_{\mu\nu}} = 0$ and we find the Raychaudhuri equation for autoparallels
to admit the form
\begin{align*}
    \frac{\mathrm{d}}{\mathrm{d}\tau}\theta^{(g)} = &-\frac{1}{3}{\theta^{(g)}}^2
    -R\indices{_{\mu\nu}}U\indices{^\mu}U\indices{^\nu}
    -{\sigma^{(g)}}\indices{_{\mu\nu}}{\sigma^{(g)}}\indices{^{\mu\nu}}\\
    &- \frac{2}{3}\theta^{(g)}S\indices{_\mu}U\indices{^\mu}
    -2\frac{\mathrm{D}}{\mathrm{D}\tau}(S\indices{_\mu}U\indices{^\mu}). \numberthis
\end{align*}
In order to render $\mathrm{d}\theta^{(g)}/\mathrm{d}\tau$
negative and hence establish the formation of conjugate points, 
the modified SEC
\begin{equation}
    R\indices{_{\mu\nu}}\xi\indices{^\mu}\xi\indices{^\nu} \geq - \frac{2}{3}\theta^{(g)}S\indices{_\mu}\xi\indices{^\mu}
    -2\frac{\mathrm{D}}{\mathrm{D}\tau}(S\indices{_\mu}\xi\indices{^\mu})
    \label{modSEC_vect_tors}
\end{equation}
for any timelike $\xi\indices{^\mu}$ has to be assumed.
Then indeed, since ${\sigma^{(g)}}\indices{_{\mu\nu}}{\sigma^{(g)}}\indices{^{\mu\nu}} \geq 0$ due to ${\sigma^{(g)}}\indices{_{\mu\nu}}$ being spacelike,
we have that $\mathrm{d}\theta^{(g)}/\mathrm{d}\tau \leq -{\theta^{(g)}}^2/3$. Therefore we have:

\begin{proposition}
Let $(\mathcal{M},g)$ be a spacetime, equipped with an independent
metric compatible affine connection $\Gamma$. Let the Ricci tensor satisfy \eqref{modSEC_vect_tors}
for all timelike $\xi\indices{^\mu}$ and let the torsion tensor $S\indices{_{\mu\nu\lambda}}$ be of form \eqref{vector_tors}. Consider a hypersurface orthogonal congruence
of timelike autoparallels emanating from $p \in \mathcal{M}$.
If $\theta_0\coloneqq \theta(\tau_0) < 0$ at some 
point $r\in\mathcal{M}$ at parameter value $\tau_0$ of a member of the congruence in the future 
of $p$, then there exists a point $q$ conjugate to $p$ within
proper time $\tau\leq -3/\theta_0$ from $r$ (assuming the congruence extends that far).
\end{proposition}

The proof proceeds exactly as in the previous subsection.

In \cite{Kranas:2018jdc}, a particular realisation of this vectorial torsion is used.
More specifically
\begin{equation}
    S\indices{_{\mu\nu\lambda}} = 2\phi g\indices{_{\mu[\nu}}U\indices{_{\lambda]}},
    \label{Kranas_tors}
\end{equation}
where $\phi = \phi(t)$ is a scalar function that depends only
on cosmic time and $g\indices{_{\mu\nu}}$ is the Friedmann-Lemaître-Robertson-Walker (FLRW) metric.
Therefore ${S\indices{_\mu} = -3\phi U\indices{_\mu}}$. Since \cite{Kranas:2018jdc} studies
behaviours of the background cosmology, the Cosmological Principle is required to hold.
The most general allowed values of the
torsion tensor in a spacetime compatible with the Cosmological
Principle are provided in \cite{TSAMPARLIS197927}. Indeed, \eqref{Kranas_tors} is a subset of these
allowed values if $U\indices{_{i}} = 0$ for $i=1,2,3$.
This ensures that $S\indices{_{00i}} = 0$ for $i=1,2,3$, as
desired. Since the vector field $U\indices{^\mu}$ is assumed
to be normalised, we conclude that
$U\indices{^\mu} = (1,0,0,0)^\intercal$. Then indeed
$U\indices{^\mu}U\indices{_\mu}=-1$ and one may verify that
$U\indices{^\alpha}\nabla\indices{_\alpha}U\indices{^\mu} = 0$,
which means that $U\indices{^\mu}$ corresponds to an autoparallel.

The Raychaudhuri equation reduces to
\begin{align*}
    \frac{\mathrm{d}}{\mathrm{d}t}\theta^{(g)} = &-\frac{1}{3}{\theta^{(g)}}^2
    -R\indices{_{00}}
    -{\sigma^{(g)}}\indices{_{\mu\nu}}{\sigma^{(g)}}\indices{^{\mu\nu}}\\
    &- 2\theta^{(g)} \phi(t)
    -6\frac{\mathrm{d}}{\mathrm{d}t}\phi(t), \numberthis
\end{align*}
where the derivative is now with respect to cosmic time
since $U\indices{^\mu} = (1,0,0,0)^\intercal$ and thus
$\mathrm{d}/\mathrm{d}\tau=\mathrm{d}/\mathrm{d}t$.
Convergence of the congruence is hence ensured if 
\begin{equation}
    R\indices{_{00}} \geq - 2\theta^{(g)} \phi(t)
    -6\frac{\mathrm{d}}{\mathrm{d}t}\phi(t).
\end{equation}
Decomposing 
\begin{equation}
    R\indices{_{00}} = \mathring{R}\indices{_{00}} -
    6\frac{\mathrm{d}}{\mathrm{d}t}\phi
\end{equation}
with $\mathring{R}\indices{_{00}} = -3\ddot{a}/a$ 
and using $\theta^{(g)} = \theta - 2U\indices{^\alpha}S\indices{_\alpha} = \theta - 6\phi$
with $\theta = 3\dot{a}/a + 6\phi$,
we find $\theta^{(g)} = 3\dot{a}/a$ and
restriction on the Universe's expansion
\begin{equation}
    \ddot{a} \leq  2\dot{a}\phi(t).
\end{equation}
Implications of this restriction would be interesting
to explore in future work.

\subsection{Weyssenhoff fluid}
The Weyssenhoff fluid is a perfect fluid with spin. 
Macroscopically it is a continuous medium, however, 
microscopically the Weyssenhoff fluid is characterised by the spin of matter fields \cite{Brechet:2007cj}.
It is mainly studied within the context of ECSK
theory, where it serves as the source of torsion \cite{Bohmer:2017dqs,Hehl1,YNObukhov_1987,Brechet:2008zz,Szydlowski:2003nv}. In ECSK, torsion is related algebraically to the
spin of matter via
\begin{equation}
    S\indices{^\lambda_{\mu\nu}} + 2\delta^\lambda_{[\mu}S\indices{_{\nu]}} = 
    \kappa \Sigma\indices{^\lambda_{\mu\nu}},
    \label{EC_tors}
\end{equation}
where $\kappa\coloneqq 8\pi G/c^4$ and $\Sigma\indices{^\lambda_{\mu\nu}}$ is the spin-density tensor.

It was postulated in \cite{YNObukhov_1987} that for a 
Weyssenhoff fluid, the source $\Sigma\indices{^\lambda_{\mu\nu}}$ admits the form
\begin{equation}
    \Sigma\indices{^\lambda_{\mu\nu}} = U\indices{^\lambda}
    \Sigma\indices{_{\mu\nu}},
    \label{Weyss_ans}
\end{equation}
where $U\indices{^\lambda}$ denotes the 4-velocity of the
fluid and $\Sigma\indices{_{\mu\nu}}$ is an antisymmetric
tensor describing the spin-density of matter. 
Furthermore, it was shown in \cite{YNObukhov_1987} that
the Frenkel condition 
\begin{equation}
    U\indices{^\mu}
    \Sigma\indices{_{\mu\nu}} = 0
\end{equation}
necessarily arises, meaning that the spin of the matter fields
is spacelike, hence $\Sigma\indices{_{\mu\nu}}\Sigma\indices{^{\mu\nu}} \geq 0$.
Inserting the ansatz \eqref{Weyss_ans} into \eqref{EC_tors},
and contracting over $\lambda$ and $\nu$, we find that
the torsion vector vanishes $S\indices{_\mu} = 0$.
Therefore the torsion tensor finally reads
\begin{equation}
    S\indices{^\lambda_{\mu\nu}} = \kappa 
    U\indices{^\lambda}
    \Sigma\indices{_{\mu\nu}}
    \label{Weyss_tors}
\end{equation}
for a Weyssenhoff fluid in ECSK.
For simplicity we set $\kappa = 1$ in the following (or
likewise we could absorb $\kappa$ into $\Sigma\indices{_{\mu\nu}}$).

For a Weyssenhoff fluid we find that ${\theta^{(g)} = 
\theta}$, ${{\omega^{(g)}}\indices{_{\mu\nu}} = \omega\indices{_{\mu\nu}}}$, and ${\sigma^{(g)}}\indices{_{\mu\nu}} = \sigma\indices{_{\mu\nu}}$.
Then the Raychaudhuri equation for autoparallels becomes
\begin{equation}
    \frac{\mathrm{d}}{\mathrm{d}\tau}\theta = -\frac{1}{3}\theta^2
    -R\indices{_{\mu\nu}}U\indices{^\mu}U\indices{^\nu}
    -\sigma\indices{_{\mu\nu}}\sigma\indices{^{\mu\nu}}
    +\omega\indices{_{\mu\nu}}\omega\indices{^{\mu\nu}},
    \label{Weyss_Raychaud}
\end{equation}
where the vorticity, assuming hypersurface orthogonality,
can be shown to exactly match the spin-density
${\omega\indices{_{\mu\nu}} = \Sigma\indices{_{\mu\nu}}}$.
In order to establish a similar proposition as in the previous sections we would now have to require the modified SEC, i.e., that
\begin{equation}
    R\indices{_{\mu\nu}}\xi\indices{^\mu}\xi\indices{^\nu} \geq 
    \Sigma\indices{_{\mu\nu}}\Sigma\indices{^{\mu\nu}}
    \label{mod_SEC_Weyss}
\end{equation}
holds for all timelike vectors $\xi\indices{^\mu}$.

\begin{proposition}
Let $(\mathcal{M},g)$ be a spacetime, equipped with an independent
metric compatible affine connection $\Gamma$.
Let matter be described by a Weyssenhoff fluid in EC theory,
i.e., torsion admits the form $S\indices{^\lambda_{\mu\nu}} = \kappa U\indices{^\lambda}\Sigma\indices{_{\mu\nu}}$, where 
$\kappa\coloneqq 8\pi G/c^4$, $U\indices{^\lambda}$ is the 4-velocity of the fluid and
$\Sigma\indices{_{\mu\nu}}$ an antisymmetric tensor describing
the spin-density of matter.
Let the Ricci tensor satisfy \eqref{mod_SEC_Weyss}
for all timelike $\xi\indices{^\mu}$. Consider a hypersurface orthogonal congruence
of timelike autoparallels emanating from $p \in \mathcal{M}$.
If $\theta_0\coloneqq \theta(\tau_0) < 0$ at some 
point $r\in\mathcal{M}$ at parameter value $\tau_0$ of a member of the congruence in the future 
of $p$, then there exists a point $q$ conjugate to $p$ within
proper time $\tau\leq -3/\theta_0$ from $r$ (assuming the congruence extends that far).
\end{proposition}

The shear equation is found to be 
\begin{align*}
    \frac{\mathrm{D}}{\mathrm{D}\tau}\sigma\indices{_{\mu\nu}} = 
    & R\indices{_{\langle\mu|\alpha\beta|\nu\rangle}}U\indices{^\alpha}
    U\indices{^\beta} - \frac{2}{3}\theta \sigma\indices{_{\mu\nu}}
    - \sigma\indices{_{\alpha\langle\mu}}
    \sigma\indices{_{\nu\rangle}^\alpha}\\
    &- \Sigma\indices{_{\alpha\langle\mu}}
    \Sigma\indices{_{\nu\rangle}^\alpha}\numberthis
\end{align*}
or likewise 
\begin{align*}
    \frac{\mathring{\mathrm{D}}}{\mathrm{D}\tau}\sigma\indices{_{\mu\nu}} = 
    & \mathring{R}\indices{_{\langle\mu|\alpha\beta|\nu\rangle}}U\indices{^\alpha}
    U\indices{^\beta} - \frac{2}{3}\theta \sigma\indices{_{\mu\nu}}
    - \sigma\indices{_{\alpha\langle\mu}}
    \sigma\indices{_{\nu\rangle}^\alpha}\\
    &- 2U\indices{^\alpha}U\indices{^\beta}U\indices{_{(\mu|}}
    \partial\indices{_{\beta}}\Sigma\indices{_{|\nu)\alpha}}
    + 2U\indices{^\alpha}U\indices{^\beta}\mathring{\Gamma}
    \indices{^\lambda_{\alpha\beta}}U\indices{_{(\mu}}
    \Sigma\indices{_{\nu)\lambda}}\\
    &- 2U\indices{^\alpha}\partial\indices{_{\langle\mu}}
    \Sigma\indices{_{\nu\rangle\alpha}}
    +2U\indices{^\alpha}\mathring{\Gamma}\indices{^\lambda
    _{\alpha\langle\nu}}\Sigma\indices{_{\mu\rangle\lambda}}
    +2\Sigma\indices{^\lambda_{(\mu}}\sigma\indices{_{\nu)\lambda}}\numberthis
\end{align*}
after splitting the Levi-Civita contributions and torsional parts,
with $\sigma\indices{_{\mu\nu}} = \mathring{\sigma}\indices{_{\mu\nu}}$.
Torsion is therefore anticipated to induce direction-dependent effects on the congruence of curves.
The vorticity equation
on the other hand reduces to an algebraic equation for the
spin-density of matter
\begin{equation}
2\Sigma\indices{_{[\nu|\alpha}}\mathring{\nabla}\indices{_{|\mu]}}U\indices{^\alpha}=
    -\Sigma\indices{^\lambda_{[\mu}}
    \Sigma\indices{_{\nu]\lambda}} - \frac{2}{3}\theta
    \Sigma\indices{_{\mu\nu}}
    + 2\sigma\indices{_{\alpha[\mu}}\Sigma\indices{_{\nu]}^\alpha},
\end{equation}
or equivalently
\begin{equation}
3 A\indices{_{\rho}}\Sigma\indices{_{[\mu}^\rho}U\indices{_{\nu]}} =
    \Sigma\indices{_\mu^\rho}\nabla\indices{_{[\nu}}U\indices{_{\rho]}}
    - \Sigma\indices{_\nu^\rho}\nabla\indices{_{[\mu}}U\indices{_{\rho]}}.
    \label{vorticity_eq_weyss}
\end{equation}
At first glance, \eqref{vorticity_eq_weyss} appears to impose a different constraint than \eqref{EC_tors}. Investigating the consistency of these two equations is necessary, though it falls outside the scope of this work.

\section{Singularity Theorem with Torsion}
\label{Sing_thm_sec}
In order to be able to prove the generalised Singularity
Theorem below for timelike curves with torsion, we need 
to state Wald's proposition \cite[p.\,226]{wald:1984}
in the most general case.
Define the scalar 
\begin{align*}
    C(U\indices{^\mu})\coloneqq \,\,&{\omega^{(g)}}\indices{_{\mu\nu}}{\omega^{(g)}}\indices{^{\mu\nu}}
    + \mathrm{D}\indices{_{\mu}}A\indices{^\mu} + 
    A\indices{_\mu}A\indices{^\mu}\\
    &- \frac{2}{3}\theta^{(g)} 
    S\indices{_\mu}U\indices{^\mu}
    - 2 S\indices{_{\mu\nu\rho}}
    U\indices{^\mu}U\indices{^\nu}A\indices{^\rho}\\
    &+ 2 S\indices{_{\mu\nu\rho}}
    {\sigma^{(g)}}\indices{^{\mu\nu}}U\indices{^\rho}
    -
    2 S\indices{_{\mu\nu\rho}}
    {\omega^{(g)}}\indices{^{\mu\nu}}U\indices{^\rho}\\
    &-2\frac{\mathrm{D}}{\mathrm{D}\tau}\left(
    U\indices{^\alpha}S\indices{_\alpha}\right), \numberthis
\end{align*}
then the Raychaudhuri eq. \eqref{Raychaudhuri_eq}
takes the more tractable form
\begin{equation}
        \frac{\mathrm{d}}{\mathrm{d}\tau}\theta^{(g)} = 
        -\frac{1}{3}{\theta^{(g)}}^2
        -R\indices{_{\mu\nu}}U\indices{^\mu}U\indices{^\nu}
        -{\sigma^{(g)}}\indices{_{\mu\nu}}{\sigma^{(g)}}\indices{^{\mu\nu}} + C(U\indices{^\mu}).
\end{equation}
In this general formulation we have:
\begin{proposition}
Let $(\mathcal{M},g)$ be a spacetime, equipped with an independent
metric compatible affine connection $\Gamma$. Let the Ricci tensor satisfy $R\indices{_{\mu\nu}}\xi\indices{^\mu}\xi\indices{^\nu} \geq C(\xi\indices{^\mu})$
for all timelike $\xi\indices{^\mu}$. Consider a hypersurface orthogonal congruence
of timelike curves emanating from $p \in \mathcal{M}$.
If $\theta_0\coloneqq \theta^{(g)}(\tau_0) < 0$ at some 
point $r\in\mathcal{M}$ at parameter value $\tau_0$ of a member of the congruence in the future 
of $p$, then there exists a point $q$ conjugate to $p$ within
proper time $\tau\leq -3/\theta_0$ from $r$ (assuming the congruence extends that far).
\label{propos_general}
\end{proposition}
Note that we are not restricting ourselves to 
autoparallel curves here. 

Another crucial point is that singularity theorems usually have a requirement on causality. In our case
that will be the condition of spacetime being \textit{globally hyperbolic}. This means that spacetime has a Cauchy surface, which per definition is a closed achronal set $\Sigma$
for which $D(\Sigma) = \mathcal{M}$. Achronal means that 
no point in $\Sigma$ belongs to the chronological future of any other point in $\Sigma$, while $D(\Sigma)$ is the domain of dependence of $\Sigma$ which contains roughly speaking all
points in spacetime that are causally connected to $\Sigma$.

The Singularity Theorem for timelike
curves with torsion thus reads:

\begin{theorem*}
Let $(\mathcal{M},g)$ be a globally hyperbolic spacetime, equipped with an independent
metric compatible affine connection $\Gamma$. Let the Ricci tensor satisfy ${R\indices{_{\mu\nu}}\xi\indices{^\mu}\xi\indices{^\nu} \geq C(\xi\indices{^\mu})}$
for all timelike $\xi\indices{^\mu}$.
Assume there exists a smooth spacelike Cauchy surface $\Sigma$
on which $\theta^{(g)}\leq K < 0$ everywhere for a past directed normal timelike congruence, where $K$ is a constant. Then no past directed timelike curve
from $\Sigma$ can have length greater that $3/\abs{K}$.
\end{theorem*}

\begin{proof}
    The proof proceeds as in \cite[p.\,237]{wald:1984}, however with some subtleties. Assume there exists a
    past-directed timelike curve $\gamma$ from $\Sigma$
    with length greater than $3/\abs{K}$. Let $p$ denote
    a point on $\gamma$ lying at a distance beyond 
    $3/\abs{K}$ from $\Sigma$. Then by \cite[Thm. 9.4.5]{wald:1984} there
    exists a maximum length curve $\tilde{\gamma}$ from 
    $\Sigma$ to $p$. Certainly $\tilde{\gamma}$ must also
    have length greater than $3/\abs{K}$ and further be timelike (since only timelike curves can be homotopic to other timelike curves). But then we know \cite[Thm. 9.4.3]{wald:1984} that
    $\tilde{\gamma}$ must be a timelike geodesic with no
    conjugate point between $\Sigma$ and $p$. This however
    contradicts Proposition\,\,\autoref{propos_general}.
    Hence the assumption is wrong and $\gamma$ cannot exist.
\end{proof}

Two remarks are in order.
First, in the presence of a general torsion tensor, we are compelled to consider non-autoparallel curves, i.e., $U\indices{^\alpha}\nabla\indices{_\alpha}
U\indices{^\mu} \neq 0$, during the step of connecting
$\Sigma$ and $p$ with a maximal length curve. The connecting
curve is a geodesic but should in general not be considered 
an autoparallel, since that would imply torsion
to be totally antisymmetric (recall that the totally antisymmetric
torsion case has been shown to be equivalent to the standard
case with no torsional contributions).
Second, the deviation of the case above to the standard case
without torsion is essentially encoded within $C(\xi\indices{^\mu})$ (apart from the hidden dependence
of the Ricci tensor on torsion). 
For the three cases studied previously, $C$
admits the following forms:
\begin{itemize}
    \item Totally antisymmetric torsion:
        \begin{equation}
            C(U\indices{^\mu}) = - 
            S\indices{^\alpha_{\mu\nu}}S\indices{^\beta^{\mu\nu}}
            U\indices{_\alpha}U\indices{_\beta}
        \end{equation}
    \item Vectorial torsion:
        \begin{equation}
            C(U\indices{^\mu}) = - \frac{2}{3}\theta^{(g)}S\indices{_\mu}U\indices{^\mu}
            -2\frac{\mathrm{D}}{\mathrm{D}\tau}(S\indices{_\mu}U\indices{^\mu})
        \end{equation}
    \item  Weyssenhoff fluid:
        \begin{equation}
            C(U\indices{^\mu}) = \Sigma\indices{_{\mu\nu}}\Sigma\indices{^{\mu\nu}}
        \end{equation}
\end{itemize}

The physical interpretation of the modified SEC, ${R\indices{_{\mu\nu}}\xi\indices{^\mu}\xi\indices{^\nu} \geq C(\xi\indices{^\mu})}$ in the theorem, invites further examination. As currently stated, this condition's implication is restricted to the convergence of timelike curves. Broader physical implications concerning the nature of matter cannot be deduced from such a general assumption. To extract such implications, it is essential to impose a specific set of gravitational field equations. This is particularly evident in the case of a Weyssenhoff fluid, where the underlying theory is assumed to be the ECSK theory. In that scenario, we were able to relate torsion to the spin-density of matter, thereby translating the modified SEC into a physically meaningful condition—specifically, that the gravitational attraction at a given point \cite{Curiel:2014zba} must be greater than or equal to the spin-density at that point. Moreover, in this model, if the spin-density is sufficiently large to violate the modified SEC, the formation of singularities is prevented. In this sense, torsion can be understood as exerting a counterbalancing effect that inhibits the development of singularities.

\section{Conclusion}
\label{sect:Concl}
This work established the foundation for formulating and ultimately proving a singularity theorem for timelike curves with torsion. Although the Raychaudhuri equation incorporating torsion has been already obtained in the literature, we presented an alternative approach by initially deriving the general deviation equation for timelike curves.
The deviation equation furthermore allowed us to generalise the concept of conjugate points, and to connect the appearance of conjugate points to the behaviour of the scalar 
expansion. 
Based on compatibility with standard cosmology, we opted to impose the condition of hypersurface orthogonality.
We observed that this condition necessitates the presence of a non-vanishing vorticity tensor, which is in stark contrast to the standard case without torsion.
For several specific torsion forms, we subsequently established modified SECs under which the emergence of conjugate points is assured. In particular we found that 
in the case of a totally antisymmetric torsion tensor, all
torsional contributions cancel out.
For the case of a Weyssenhoff fluid we discovered that 
vorticity is exactly equal to the spin-density of matter, and consequently to torsion, which facilitates the general consensus that torsion 
constitutes a twisting of spacetime.
Finally we formulated and proved a singularity theorem for 
timelike curves with torsion. We observed that in the case where torsion is not totally antisymmetric, the inclusion of non-autoparallel curves is essential for incorporating the contributions of torsion. In other words, the autoparallel acceleration terms involving $A\indices{^\mu}$ must be integrated into the Raychaudhuri equation.

While the formulation of singularity theorems represents an essential initial step, the ultimate objective is to circumvent the conditions imposed by these theorems to prevent the occurrence of singularities. The scalar $C$ introduced in the previous section fundamentally indicates how torsion can violate the modified SECs and, consequently, which torsional forms are physically feasible.
We anticipate a more thorough investigation of this matter in future research.

Furthermore, we would like to emphasise that the status of ECs remains uncertain. While they serve as useful guidelines for physically reasonable scenarios, it is evidently necessary to violate them in certain cases to prevent unphysical behavior. Thus, a reevaluation of the standing and interpretation of these ECs is warranted, as highlighted in 
\cite{Curiel:2014zba}.

\section*{Acknowledgements}
The authors wish to thank Jürgen Struckmeier, Johannes Kirsch, Vladimir Denk, Marcelo Netz-Marzola, David Benisty and Tomoi Koide for valuable discussions.
Special thanks are extended to the ``Walter Greiner-Gesell\-schaft zur F\"{o}rderung
der physi\-ka\-lischen Grundlagenforschung e.V.'' (WGG) in Frankfurt for their support.\,
AvdV and DV gratefully acknowledge the support provided by the Fueck Stiftung.

\appendix

\section{Extending the Deviation Vector Field}
\label{sect_flow_calc}
This section elaborates on the intermediate steps between  
\eqref{ext_def} and \eqref{ext_res}. First act with \eqref{ext_def}
on an arbitrary smooth function $f$ on $\mathcal{M}$.
Then, per definition of the pushforward, the r.h.s. yields
\begin{equation}
    \left((\phi_{\tau})_* \left({\hat{V}}_{\gamma(s)}\right)\right)(f) = 
    {\hat{V}}_{\gamma(s)}(f\circ\phi_\tau).
\end{equation}
Choose any local frame $x$, then we may write
\begin{equation}
    {\hat{V}}_{\gamma(s)}(f\circ\phi_\tau) = 
    {\hat{V}}_{\gamma(s)}^\mu \left(\frac{\partial}{\partial x^\mu}\right)_{\gamma(s)}(f\circ\phi_\tau),
\end{equation}
where $ \left(\partial/\partial x^\mu\right)_{\gamma(s)}$ denotes the coordinate basis of the
tangent space at $\gamma(s)$. The action of the coordinate
basis is given by
\begin{equation}
    \left(\frac{\partial}{\partial x^\mu}\right)_{\gamma(s)}(f\circ\phi_\tau) = \frac{\partial}{\partial x^\mu}
    \left.\left((f\circ \phi_\tau)\circ x^{-1}\right)\right\rvert_{x(\gamma(s))}.
    \label{AppA_helpeq1}
\end{equation}
Now choose another frame $y$, rewrite 
$(f\circ \phi_\tau)\circ x^{-1} = f\circ y^{-1} \circ( y\circ \phi_\tau\circ x^{-1})$ and use the chain rule to find
that the r.h.s. of \eqref{AppA_helpeq1} becomes
\begin{equation}
    \left.\frac{\partial (f\circ y^{-1})}{\partial y^\nu}\right\rvert_{y(\gamma_s(\tau))}
    \left.\frac{\partial(y^\nu \circ \phi_\tau \circ x^{-1})}{\partial x^\mu}\right\rvert_{x(\gamma(s))},
\end{equation}
which is equal to 
\begin{equation}
    \left.\frac{\partial(y^\nu \circ \phi_\tau \circ x^{-1})}{\partial x^\mu}\right\rvert_{x(\gamma(s))}
    \left(\frac{\partial}{\partial y^\nu}\right)_{\gamma_s(\tau)}(f).
    \label{AppA_helpeq2}
\end{equation}
The contraction of \eqref{AppA_helpeq2} with
\begin{equation}
    {\hat{V}}_{\gamma(s)}^\mu \overset{\text{def}}{=} \frac{\mathrm{d} {\gamma}\indices{^{\mu}}}{\mathrm{d} s}(s) = \frac{\mathrm{d} {(x^\mu \circ\gamma)}}{\mathrm{d} s}(s),
\end{equation}
and the use of the chain rule yield
\begin{equation}
    \frac{\mathrm{d}}{\mathrm{d}s}\left.\left(y^\nu \circ \phi_\tau \circ x^{-1}\circ x \circ \gamma\right)\right\rvert_{s}\left(\frac{\partial}{\partial y^\nu}\right)_{\gamma_s(\tau)}(f),
\end{equation}
or 
\begin{equation}
    \frac{\partial}{\partial s}\left.\left(y^\nu \circ  \gamma_s (\tau)\right)\right\rvert_{s}\left(\frac{\partial}{\partial y^\nu}\right)_{\gamma_s(\tau)}(f),
\end{equation}
which is equivalent to
\begin{equation}
    \frac{\partial {\gamma_s}\indices{^{\nu}}(\tau)}{\partial s}\left(\frac{\partial}{\partial y^\nu}\right)_{\gamma_s(\tau)}(f).
    \label{AppA_helpeq3}
\end{equation}
Comparing \eqref{AppA_helpeq3} with the l.h.s. of \eqref{ext_def}:
\begin{equation}
    {\tilde{V}}_{\gamma_s (\tau)}^\nu\left(\frac{\partial}{\partial y^\nu}\right)_{\gamma_s(\tau)}(f),
\end{equation}
we finally find
\begin{equation}
    {\tilde{V}}_{\gamma_s (\tau)}^\mu = 
    \frac{\partial {\gamma_s}\indices{^{\mu}}(\tau)}{\partial s}.
\end{equation}

\section{Energy Conditions}
\label{sect_en_cond}
Energy conditions are generically posed, independent of details of matter source, conditions on the behaviour of matter on a dynamic curved space. Since these conditions are independent of the idiosyncrasies of matter fields, these conditions typically, in conjunction with Einstein's field equations, lead to far-reaching consequences in determining the global structure of spacetime manifolds and the nature of the gravitational field, for example, through singularity theorems. If one assumes the energy-momentum tensor of the matter fields to be of Hawking Ellis type I \cite{hawking1975large} and, further, utilises the Einstein field equations, then it is possible to propose these conditions in various forms, see Table\,\ref{EC_table}. 
\begin{table*}
\caption{Collection of the most prominent ECs in their geometric, physical, and effective form \cite{Curiel:2014zba}. From top to bottom we listed the null-, weak-, strong-, dominant-, and strengthened dominant EC. The vector $k\indices{^a}$ always refers to an arbitrary null vector, whereas $\xi\indices{^a}$ and $\eta\indices{^a}$ are arbitrary timelike vectors. The effective forms hold for all $i=1,2,3$}
\begin{ruledtabular}
\begin{tabular}{l|c|c|c}
& \textbf{Geometric} & \textbf{Physical} & \textbf{Effective}\\
\hline\hline
NEC & $R\indices{_{ab}} k\indices{^a} k\indices{^b} \geq 0$  & $T\indices{_{ab}} k\indices{^a} k\indices{^b} \geq 0$ & $\rho + p_i \geq 0$\\\hphantom{.} &\hphantom{.}&\hphantom{.}&\hphantom{.}\\ 
WEC & $G\indices{_{ab}} \xi\indices{^a} \xi\indices{^b} \geq 0$ & $T\indices{_{ab}} \xi\indices{^a} \xi\indices{^b} \geq 0$ & $\rho \geq 0, \,\,\rho + p_i \geq 0$\\\hphantom{.} &\hphantom{.}&\hphantom{.}&\hphantom{.} \\
SEC & $R\indices{_{ab}} \xi\indices{^a} \xi\indices{^b} \geq 0$ & $(T\indices{_{ab}} - \frac{1}{2}T g\indices{_{ab}}) \xi\indices{^a} \xi\indices{^b} \geq 0$ & \makecell{$\rho + \sum_{j}p_j\geq 0$, \vspace*{2mm}\\ $\rho + p_i \geq 0$}\\\hphantom{.} &\hphantom{.}&\hphantom{.}&\hphantom{.} \\
DEC & \makecell{$G\indices{_{ab}} \xi\indices{^a} \xi\indices{^b} \geq 0, \,\, G\indices{^a_{b}} \xi\indices{^b}$ is causal, \vspace*{2mm}\\ $G\indices{_{ab}} \xi\indices{^a} \eta\indices{^b} \geq 0$ for
$\xi\indices{^a}, \eta\indices{^a}$ co-oriented} & \makecell{$T\indices{_{ab}} \xi\indices{^a} \xi\indices{^b} \geq 0, \,\, T\indices{^a_{b}} \xi\indices{^b}$ is causal, \vspace*{2mm}\\ $T\indices{_{ab}} \xi\indices{^a} \eta\indices{^b} \geq 0$ for
$\xi\indices{^a}, \eta\indices{^a}$ co-oriented}& $\rho \geq 0, \,\, \abs{p_i}\leq \rho$\\\hphantom{.} &\hphantom{.}&\hphantom{.}&\hphantom{.} \\
SDEC & \makecell{$G\indices{_{ab}} \xi\indices{^a} \xi\indices{^b} \geq 0, \,\, G\indices{^a_{b}} \xi\indices{^b}$ is timelike if $R\indices{_{ab}} \neq 0$, \vspace*{2mm}\\ $G\indices{_{ab}} \xi\indices{^a} \eta\indices{^b} > 0$ for
$\xi\indices{^a}, \eta\indices{^a}$ co-oriented or $G\indices{_{ab}} = 0$} & \makecell{$T\indices{_{ab}} \xi\indices{^a} \xi\indices{^b} \geq 0, \,\, T\indices{^a_{b}} \xi\indices{^b}$ is timelike if $T\indices{_{ab}} \neq 0$, \vspace*{2mm}\\ $T\indices{_{ab}} \xi\indices{^a} \eta\indices{^b} > 0$ for
$\xi\indices{^a}, \eta\indices{^a}$ co-oriented or $T\indices{_{ab}} = 0$}  & $\rho \geq 0, \,\, \abs{p_i}\leq \rho$
\end{tabular}
\end{ruledtabular}
\label{EC_table}
\end{table*}
One can propose a form involving the curvature of spacetime (geometric form), energy-momentum tensor (physical form), or the energy density and pressure of the matter fields (effective form) \cite{Curiel:2014zba}. However, it is difficult to suggest one form or the other to be fundamental. For example, the weak energy condition only has a clear physical justification through its physical form. On the other hand, one may suggest that only covariant (observer-independent) formulations of the conditions should be considered fundamental.


\bibliography{biblio}

\end{document}